\documentclass[11pt,english,letter]{article}

\usepackage{latexsym}
\usepackage{amssymb}
\usepackage{amsmath}
\usepackage{amsthm}
\usepackage{graphicx}
\usepackage{amsfonts}
\usepackage{array}
\usepackage{float}
\usepackage{algorithm} 
\usepackage{algorithmic}
\usepackage{tipa}
\usepackage{xspace}
\usepackage{color}
\newtheorem{theorem}{Theorem}%[section]

\newtheorem{definition}[theorem]{Definition}
\newtheorem{lemma}[theorem]{Lemma}

\newtheorem{remark}[theorem]{Remark}

\newtheorem{proposition}[theorem]{Proposition}
\newtheorem{claim}[theorem]{Claim}

\floatstyle{ruled}

\newcommand{\ie}{{\em i.e.,}\xspace}
\newcommand{\eg}{{\em e.g.,}\xspace}
\newcommand{\vs}{{\em vs.}\xspace}

\newcommand{\pa}{P}
\newcommand{\di}{Seed}
\newcommand{\dip}{x}
\newcommand{\rd}{Repeat\-Seed}
\newcommand{\fl}{Berry}
\newcommand{\flp}{y}
\newcommand{\bo}{Cloud\-berry}
\newcommand{\bop}{z}
\newcommand{\bops}{h}
\newcommand{\p}{Harvest}
\newcommand{\ite}{Assum\-ption}
\newcommand{\pp}{Push\-Pattern}
\newcommand{\aud}{RV}
\newcommand{\dist}{D}
\newcommand{\sholab}{l}
\newcommand{\difbit}{l'}
\newif\ifSODA
\newif\ifTR
\SODAfalse
\TRfalse

\SODAfalse
\TRtrue

\begin{document}

	\title{{\bf Asynchronous approach in the plane: A deterministic polynomial algorithm}\footnote{A preliminary version of this paper appeared in the Proceedings of the 31st International Symposium on Distributed Computing (DISC 2017). This work was performed within Project ESTATE (Ref. ANR-16-CE25-0009-03) and Project TOREDY. The first project is supported by French state funds managed by the ANR (Agence Nationale de la Recherche), while the second project is supported by the European Regional Development Fund (ERDF) and the Hauts-de-France region.}}

	\author{S\'ebastien Bouchard$^{1}$, %
	Marjorie Bournat$^{1}$, %
	Yoann Dieudonn\'e$^{2}$, %
	Swan Dubois$^{1}$, %
	Franck Petit$^1$\\
	\small$^1$ Sorbonne Universit\'e, CNRS, INRIA, LIP6, F-75005 Paris, France\\
	\small$^2$  MIS Lab., Universit\'e de Picardie Jules Verne, France
	}
	%\author{\IEEEauthorblockN{S\'ebastien Bouchard\IEEEauthorrefmark{1},
	%Marjorie BOURNAT\IEEEauthorrefmark{1},
	%Yoann DIEUDONNE\IEEEauthorrefmark{2},
	%Swan DUBOIS\IEEEauthorrefmark{1}, and
	%Franck PETIT\IEEEauthorrefmark{1}
	%\IEEEauthorblockA{\IEEEauthorrefmark{1}Sorbonne Universit\'es, UPMC Univ Paris 06,\\
	%CNRS, Inria, LIP6\\
	%F-75005, Paris, France\\
	%Email: firstname.lastname@lip6.fr}
	%}
	%\IEEEauthorblockA{\IEEEauthorrefmark{2} Universit\'e de Picardie Jules Verne, \\
	%MIS \\
	%F-80000, Amiens, France\\
	%Email: yoann.dieudonne@u-picardie.fr }
	%}
	%
	\date{\vspace{-5ex}}

	\maketitle

\begin{abstract}
		In this paper we study the task of \emph{approach} of two mobile agents having the same limited range
		of vision and moving asynchronously in the plane. This task consists in getting them in finite time
		within each other's range of vision. The agents execute the same deterministic algorithm and are
		assumed to have a compass showing the cardinal directions as well as a unit measure. On the other
		hand, they do not share any global coordinates system (like GPS), cannot communicate and have distinct
		labels. Each agent knows its label but does not know the label of the other agent or the initial
		position of the other agent relative to its own. The route of an agent is a sequence of segments that
		are subsequently traversed in order to achieve approach. For each agent, the computation of its route
		depends only on its algorithm and its label. An adversary chooses the initial positions of both agents
		in the plane and controls the way each of them moves along every segment of the routes, in particular
		by arbitrarily varying the speeds of the agents. Roughly speaking, the goal of the adversary is to
		prevent the agents from solving the task, or at least to ensure that the agents have covered as much
		distance as possible before seeing each other. A deterministic approach algorithm is a deterministic
		algorithm that always allows two agents with any distinct labels to solve the task of approach
		regardless of the choices and the behavior of the adversary. The cost of a complete execution of an
		approach algorithm is the length of both parts of route travelled by the agents until approach is completed.

		%OPTION 1
		Let $\Delta$ and $l$ be the initial distance separating the agents and the length of (the binary representation of) the shortest label,
		respectively.  {\it Assuming that $\Delta$ and $l$ are unknown to both agents, does there exist a deterministic 
		approach algorithm always working at a cost that is polynomial in $\Delta$ and $l$?}

		%OPTION 2
%		{\it Does there exist a deterministic approach algorithm always working at cost polynomial in the
		%		unknown initial distance separating the agents and in the length of (the binary
		%		representation of) the shortest label?}

		Actually the problem of approach in the plane reduces to the network problem of rendezvous in an infinite oriented grid, which consists in ensuring that both agents end up meeting at the same time at a node or on an edge of the grid. By designing such a rendezvous algorithm with appropriate properties, as we do in this paper, we provide a positive answer to the above question.

		Our result turns out to be an important step forward from a computational point of view, as the other algorithms allowing to solve the same problem either have an exponential cost in the initial separating distance and in the labels of the agents, or require each agent to know its starting position in a global system of coordinates, or only work under a much less powerful adversary.

		\vspace{2ex}

\end{abstract}

\section{Introduction}
\label{sec:intro}

\ifTR
  \subsection{Model and Problem}
  \label{sub:subm}
\else
  \ifSODA
    \paragraph{Model and Problem.}
  \fi
\fi

The distributed system considered in this paper consists of two {\em mobile agents} that are initially placed by an adversary at arbitrary but distinct positions in the plane. Both agents have a {\em limited sensory radius} (in the sequel also referred to as {\em radius of vision}), the value of which is denoted by $\epsilon$, allowing them to sense (or, to see) all their surroundings at distance at most $\epsilon$ from their respective current locations. We assume that the agents know the value of $\epsilon$. As stated in \cite{CzyzowiczPL12}, when $\epsilon=0$, if agents start from arbitrary positions of the plane and can freely move on it, making them occupy the same location at the same time is impossible in a deterministic way. So, we assume that $\epsilon>0$ and we consider the task of {\em approach} which consists in bringing them at distance at most $\epsilon$ so that they can see each other. In other words, the agents completed their approach once they mutually sense each other and they can even get closer.  Without loss of generality, we assume in the rest of this paper that $\epsilon=1$.

The initial positions of the agents, arbitrarily chosen by the adversary, are separated by a distance $\Delta$ that is initially unknown to both agents and that is greater than $\epsilon=1$.
In addition to the initial positions, the adversary also assigns a different non-negative integer (called label) to each agent. The label of an agent is the only input of the deterministic algorithm executed by the agent. While the labels are distinct, the algorithm is the same for both agents. Each agent is equipped with a compass showing the cardinal directions and with a unit of length. The cardinal directions and the unit of length are the same for both agents. 

To describe how and where each agent moves, we need to introduce two important notions that are borrowed from \cite{CzyzowiczPL12}: The {\em route} and the {\em walk} of an agent. The {\em route} of an agent is a sequence $(S_1,S_2,S_3\ldots)$ of segments $S_i=[a_i,a_{i+1}]$ traversed in stages as follows. The route starts from $a_1$, the initial position of the agent. For every $i\geq 1$, starting from the position $a_i$, the agent initiates Stage~$i$ by choosing a direction $\alpha$ (using its compass) as well as a distance $x$. Stage~$i$ ends as soon as the agent either sees the other agent or reaches $a_{i+1}$ corresponding to the point at distance $x$ from $a_i$ in direction $\alpha$. Stages are repeated indefinitely (until the approach is completed).
%The route of each agent depends essentially on its algorithm, its label and its initial position (the coordinates of which in a global coordinate system are unknown by the agent). 
Since both agents never know their positions in a global coordinate system, the directions they choose at each stage can only depend on their (deterministic) algorithm and their labels. So, the route (the actual sequence of segments) followed by an agent depends on its algorithm and its label, but also on its initial position.
By contrast, the {\em walk} of each agent along every segment of its route is controlled by the adversary. More precisely, within each stage $S_i$ and while the approach is not achieved, the adversary can arbitrarily vary the speed of the agent, stop it and even move it back and forth as long as the walk of the agent is continuous, does not leave $S_i$, and ends at $a_{i+1}$. Roughly speaking, the goal of the adversary is to prevent the agents from solving the task, or at least to ensure that the agents have covered as much distance as possible before seeing each other. We assume that at any time an agent can remember the route it has followed since the beginning.

A {\em deterministic approach algorithm} is a deterministic algorithm that always allows two agents to solve the task of approach regardless of the choices and the behavior of the adversary. The {\em cost} of an accomplished approach is the length of both parts of route travelled by the agents until they see each other. An approach algorithm is said to be {\em polynomial} in $\Delta$ and in the length of the binary representation of the shortest label between both agents if it always permits to solve the problem of approach at a cost that is polynomial in the two aforementioned parameters, no matter what the adversary does.

It is worth mentioning that the use of distinct labels is not fortuitous. In the absence of a way of distinguishing the agents, the task of approach would have no deterministic solution. This is especially the case if the adversary handles the agents in a perfect synchronous manner. Indeed, if the agents act synchronously and have the same label, they will always follow the same deterministic rules leading to a situation in which the agents will always be exactly at distance $\Delta$ from each other.

\ifTR
  \subsection{Our Results}
  \label{sub:res}
\else
  \ifSODA
    \paragraph{Our Results.}
  \fi
\fi

In this paper, we prove that the task of approach can be solved deterministically in the above asynchronous model, at a cost that is polynomial in the unknown initial distance separating the agents and in the length of the binary representation of the shortest label. To obtain this result, we go through the design of a deterministic algorithm for a very close problem, that of rendezvous in an infinite oriented grid which consists in ensuring that both agents end up meeting either at a node or on an edge of the grid. The tasks of approach and rendezvous are very close as the former can be reduced to the latter.

It should be noticed that our result turns out to be an important advance, from a computational point of view, in resolving the task of approach. Indeed, the other existing algorithms allowing to solve the same problem either have an exponential cost in the initial separating distance and in the labels of the agents \cite{CzyzowiczPL12}, or require each agent to know its starting position in a global system of coordinates \cite{CollinsCGL10}, or only work under a much less powerful adversary \cite{DieudonneP15} which initially assigns a possibly different speed to each agent but cannot vary it afterwards. %These elements are specified and refined in the related works.

\ifTR
  \subsection{Related Work}
  \label{sub:rw}
\else
  \ifSODA
   \vspace*{-6pt}

   \paragraph{Related Work.} 
  \fi
\fi

The task of approach is closely linked to the task of rendezvous. Historically, the first mention of the rendezvous problem appeared in \cite{Schelling}. %Rendezvous is the term which is usually used when the studied task of gathering is restricted to a team of exactly two agents. 
From this publication until now, the problem has been extensively studied so that there is henceforth a huge literature about this subject. This is mainly due to the fact that {there are a lot of alternatives} for the combinations we can make when addressing the problem, \eg playing on the environment in which the agents are supposed to evolve, the way of applying the sequences of instructions (\ie deterministic or randomized) or the ability to leave some traces in the visited locations, etc. Naturally, in this paper {we focus on the work} related to deterministic rendezvous. This is why we will mostly dwell on this scenario in the rest of this 
\ifTR
   subsection.
\else
  \ifSODA
     paragraph.
  \fi
\fi
However, for the curious reader wishing to consider the matter in greater depth, regarding randomized rendezvous, a good starting point is to go through \cite{Alpern02,Alpern03,KranakisKR06}. Concerning deterministic rendezvous, the literature is divided according to the way of modeling the {environment: agents} can either move in a graph representing a network, or in the plane.

For the problem of rendezvous in networks, a lot of papers considered synchronous settings, \ie a context where the agents move in the graph in synchronous rounds. This is particularly the case of \cite{DessmarkFKP06} in which the authors presented a deterministic protocol for solving the rendezvous problem, which guarantees a meeting of the two involved agents after a number of rounds that is polynomial in the size $n$ of the graph, the length $l$ of the shortest of the two labels and the time interval $\tau$ between their wake-up times. As an open problem, the authors asked whether it was possible to obtain a polynomial solution to this problem which would be independent of $\tau$. A positive answer to this question was given, independently of each other, in~\cite{KowalskiM08} and~\cite{Ta-ShmaZ14}. While these algorithms ensure rendezvous in polynomial time (\ie a polynomial number of rounds), they also ensure it at polynomial cost because the cost of a rendezvous protocol in a graph is the number of edges traversed by the agents until they meet---each agent can make at most one edge traversal per round. 
Note that despite the fact a polynomial time implies a polynomial cost in this context, the reciprocal is not always true as the agents can have very long waiting periods, sometimes interrupted by a movement. Thus these parameters of cost and time are not always linked to each other. This was highlighted in \cite{MillerP16} where the authors studied the tradeoffs between cost and time for the deterministic rendezvous problem.  More recently, some efforts have been dedicated to analyse the impact on time complexity of rendezvous when in every round the agents are brought with some pieces of information by making a query to some device or some oracle~\cite{DKU14,MillerP14b}. Along with the work aiming at optimizing the parameters of time and/or cost of rendezvous, some other work have examined the amount of required memory to solve the problem, \eg \cite{FraigniaudP08,FraigniaudP13} for tree networks and in \cite{CzyzowiczKP12} for general networks. In~\cite{ChalopinDLP16}, the problem is approached in a fault-prone framework, in which the adversary can delay an agent for a finite number of rounds, each time it wants to traverse an edge of the network.

Rendezvous is the term that is usually used when the task of meeting is restricted to a team of exactly two agents. When considering a team of two agents or more, the term of gathering is commonly used. Still in the context of synchronous networks, we can cite some work about gathering two or more agents. In~\cite{DieudonneP16}, the task of gathering is studied for anonymous agents while in~\cite{BouchardDD16,0001LM15,DieudonnePP14} the same task is studied in presence of byzantine agents that are, roughly speaking, malicious agents with an arbitrary behavior.

Some studies have been also dedicated to the scenario in which the agents move asynchronously in a network \cite{CzyzowiczPL12,DieudonnePV15,MarcoGKKPV06}, \ie assuming that the agent speed may vary, controlled by the adversary. In \cite{MarcoGKKPV06}, the authors investigated the cost of rendezvous for both infinite and finite graphs.  In the former case, the graph is reduced to the (infinite) line and bounds are given depending on whether the agents know the initial distance between them or not. In the latter case (finite graphs), similar bounds are given for ring shaped networks.  They also proposed a rendezvous algorithm for an arbitrary graph provided the agents initially know an upper bound on the size of the graph. This assumption was subsequently removed in \cite{CzyzowiczPL12}.  However, in both~\cite{MarcoGKKPV06} and~\cite{CzyzowiczPL12}, the cost of rendezvous was exponential in the size of the graph. The first rendezvous algorithm working for arbitrary finite connected graphs at cost polynomial in the size of the graph and in the length of the shortest label was presented in \cite{DieudonnePV15}. (It should be stressed that the algorithm from \cite{DieudonnePV15} cannot be used to obtain the solution described in the present paper: this point is fully explained in the end of this subsection). In all the aforementioned studies, the agents can remember all the actions they have made since the beginning. A different asynchronous scenario for networks was studied in \cite{DAngeloSN14}. In this paper, the authors assumed that agents are oblivious, but they can observe the whole graph and make navigation decisions based on these observations.

Concerning rendezvous or gathering in the plane, we also found the same dichotomy of synchronicity \vs asynchronicity. The synchronous case was introduced in~\cite{SuzukiY99} and studied from a fault-tolerance point of view in~\cite{AgmonP06,DefagoGMP06,DieudonneP12}. In~\cite{IzumiSKIDWY12}, rendezvous in the plane is studied for oblivious agents equipped with unreliable compasses under synchronous and asynchronous models. Asynchronous gathering of many agents in the plane has been studied in various settings in \cite{CieliebakFPS12,CohenP05,CohenP08,FlocchiniPSW05,PagliPV15}. However, the common feature of all these papers related to rendezvous or gathering in the plane -- which is not present in our model -- is that the agents can observe all the positions of the other agents or at least the global graph of visibility is always connected (\ie the team cannot be split into two groups so that no agent of the first group can detect at least one agent of the second group).

Finally, the closest works to ours allowing to solve the problem of approach under an asynchronous framework are \cite{CollinsCGL10,BampasCGIL10,CzyzowiczPL12,DieudonneP15}. In \cite{CollinsCGL10,CzyzowiczPL12,DieudonneP15}, the task of approach is solved by reducing it to the task of rendezvous in an infinite oriented grid. In \cite{BampasCGIL10}, the authors present a solution to solve the task of approach in a multidimensional space by reducing it to the task of rendezvous in an infinite multidimensional grid. Let us give some more details concerning these four works to highlight the contrasts with our present contribution. The result from \cite{CzyzowiczPL12} leads to a solution to the problem of approach in the plane but has the disadvantage of having an exponential cost. % exponential in the initial distance between agents and in the larger of the labels.
The result from \cite{CollinsCGL10} and~\cite{BampasCGIL10} also implies a solution to the problem of approach in the plane at cost polynomial in the initial distance of the agents. However, in both these works, the authors %of \cite{CollinsCGL10}
use the powerful assumption that each agent knows its starting position in a global system of coordinates (while in our paper, the agents are completely ignorant of where they are).  
Lastly, the result from \cite{DieudonneP15} provides a solution at cost polynomial in the initial distance between agents and in the length of the shortest label. However, the authors of this study also used a powerful assumption: The adversary initially assigns a possibly different and arbitrary speed to each agent but cannot vary it afterwards. Hence, each agent moves at constant speed and uses clock to achieve approach. By contrast, in our paper, we assume basic asynchronous settings, \ie the adversary arbitrarily and permanently controls the speed of each agent.

To close this subsection, it is worth mentioning that it is unlikely that the algorithm from \cite{DieudonnePV15} that we referred to above, which is especially designed for asynchronous rendez-vous in arbitrary finite graphs, could be used to obtain our present result. First, in \cite{DieudonnePV15} the algorithm has not a cost polynomial in the initial distance separating the agents and in the length of the smaller label. Actually, ensuring rendezvous at this cost is even impossible in {an arbitrary graph}, as witnessed by the case of the clique with two agents labeled $0$ and $1$: the adversary can hold one agent at a node and make the other agent traverse $\Theta(n)$ edges before rendezvous, in spite of the initial distance $1$. Moreover, the validity of the algorithm given in \cite{DieudonnePV15} closely relies on the fact that both agents must evolve in the same finite graph, which is clearly not the case in our present scenario. In particular even when considering the task of rendezvous in an infinite oriented grid, the natural attempt consisting in making each agent apply the algorithm from \cite{DieudonnePV15} within bounded grids of increasing size and centered in its initial position, does not permit to claim that rendezvous ends up occurring. Indeed, the bounded grid considered by an agent is never exactly the same than the bounded grid considered by the other one (although they may partly overlap), and thus the agents never evolve in the same finite graph which is a necessary condition to ensure the validity of the solution of \cite{DieudonnePV15} and by extension of this natural attempt.

\ifTR
  \subsection{Roadmap}

\else
  \ifSODA
\vspace*{-6pt}

\paragraph{Roadmap.} 
  \fi
\fi

The next section (Section~\ref{sec:pre}) is dedicated to the computational
model and basic definitions. 
We sketch our solution in Section~\ref{sec:idea}, %
\ifTR
formally described in Sections~\ref{sec:bas} and~\ref{sec:principalAlgorithm}.  
Section~\ref{sec:proof} presents the correctness proof and cost analysis of the algorithm. 
\else
  \ifSODA
more formally described in Sections~\ref{sec:principalAlgorithm}---details on the algorithm, the proofs of correction, and cost analysis are given in appendix. 
  \fi
\fi
Finally, we make some concluding remarks in Section~\ref{sec:ccl}.

\section{Preliminaries} \label{sec:pre}

	We know from \cite{CzyzowiczPL12,DieudonneP15} that the problem of approach in the plane can be reduced to that of rendezvous in an infinite grid specified in the next paragraph.

	Consider an {\em infinite square grid} in which every node $u$ is adjacent to $4$ nodes located North, East,
	South, and West from node $u$.  We call such a grid a \emph{basic grid}. Two agents with distinct labels
	(corresponding to non-negative integers) starting from arbitrary and distinct nodes of a basic grid $G$ have
	to meet either at some node or inside some edge of $G$. As for the problem of approach (in the plane), each agent is equipped
	with a compass showing the cardinal directions. The agents can see each other and communicate only when they
	share the same location in $G$. In other words, in the basic grid $G$ we assume that the sensory radius (or,
	radius of vision) of the agents is equal to zero. 
	In such settings, the only initial input that is given to a rendezvous algorithm is the label of the executing agent.
	When occupying a node $u$, an agent decides (according to its algorithm) to move to an adjacent node $v$ via
	one of the four cardinal directions: the movement of the agent along the edge $\{u,v\}$ is controlled by the
	adversary in the same way as in a section of a route (refer to %
\ifTR Subsection~\ref{sub:subm}\else \ifSODA Section~{\ref{sec:intro}}\fi\fi%
\xspace), \ie the adversary can arbitrarily vary the speed of the agent, stop it and even move it back and forth as long as the walk of the agent is continuous, does not leave the edge, and ends at $v$.

The {\em cost} of a rendezvous algorithm in a basic grid is the total number of edge traversals by both agents until their meeting.

%The code of a path $\pi$, denoted $code(\pi)$ from a node $u$ to a node $v$ is as follows. If path $\pi$ is empty then its code is the empty word, otherwise $\pi$ is made of an edge $\{u,w\}$ followed by a path $\pi'$ from $w$ to $v$: the code of $\pi$ is $X.code(\pi')$ where $X$ is the direction from $u$ to $w$ and the operator $.$ is the concatenation.  

	From the reduction described in \cite{DieudonneP15}, we have the following theorem.

	\begin{theorem}
	\label{theo:preli}
		If there exists a deterministic algorithm solving the problem of rendezvous between any two agents in
		a basic grid at cost polynomial in $D$ and in the length of the binary representation of the shortest
		of their labels where $D$ is the distance (in the Manhattan metric) between the two starting nodes
		occupied by the agents, then there exists a deterministic algorithm solving the problem of approach in
		the plane between any two agents at cost polynomial in $\Delta$ and in the length of the binary
		representation of the shortest of their labels where $\Delta$ is the initial Euclidean distance separating the agents.
	\end{theorem}

	For completeness let us now outline the reduction described in~\cite{DieudonneP15}.
	Consider an infinite square grid with edge length 1. More precisely, for any point $v$ in the plane, we define
	the {\em basic grid} $G_v$ to be the infinite graph, one of whose nodes is $v$, and in which every node $u$ is
	adjacent to 4 nodes at Euclidean distance 1 from it, and located North, East, South, and West from node $u$.
	We now focus on how to transform any rendezvous algorithm in the grid $G_v$ 
	%(where rendezvous means simultaneously bringing two agents starting at arbitrary nodes of the grid to the same 
	%node or to the same point inside some edge) 
	to an algorithm for the task of approach in the plane. 
	
	Let $A$ be any rendezvous algorithm for any basic
	grid. Algorithm $A$ can be executed in the grid $G_w$, for any point $w$ in the plane. Consider two agents in
	the plane starting respectively from point $v$ and from another point $w$ in the plane. 
	Let $V'$ be the set of nodes in $G_v$ that are the closest nodes from $w$. Let $v'$ be a node in $V'$,
	arbitrarily chosen.
	%Let $v'$ be the node of $G_v$ closest to point $w$. If there are more than one closest nodes, we pick one of them arbitrarily.
	Notice that $v'$ is at distance at most $\sqrt{2}/2<1$ from $w$. Let $\alpha$ be the vector $v'w$. Execute
	algorithm $A$ on the grid $G_v$ with agents starting at nodes $v$ and $v'$. Let $p$ be the point in $G_v$
	(either a node of it or a point inside an edge), in which these agents meet at some time $t$. The transformed
	algorithm $A^*$ for approach in the plane works as follows: Execute the same algorithm $A$ but with one agent
	starting at $v$ and traveling in $G_v$ and the other agent starting at $w$ and traveling in $G_w$, so that the
	starting time of the agent starting at $w$ is the same as the starting time of the agent starting at $v'$ in
	the execution of $A$ in $G_v$. The starting time of the agent starting at $v$ does not change. If approach has
	not been accomplished before, in time $t$ the agent starting at $v$ and traveling in $G_v$ will be at point
	$p$, as previously. In the same way, the agent starting at $w$ and traveling in $G_w$ will get to some point
	$q$ at time $t$. Clearly, $q=p+\alpha$. Hence both agents will be at distance less than 1 at time $t$, which
	means that they accomplish approach in the plane because $\epsilon=1$ (refer to %
\ifTR Subsection~\ref{sub:subm}\else \ifSODA Section~{\ref{sec:intro}}\fi\fi%
\xspace). 

	Hence in the rest of the paper we will consider rendezvous in a basic grid, instead of the task of approach.
	We use $N$ (resp. $E$, $S$, $W$) to denote the cardinal direction North (resp. East, South, West) and an
	instruction like ``Perform $NS$'' means that the agent traverses one edge to the North and then traverses one
	edge to the South (by the way, coming back to its initial position). We denote by $D$ the initial (Manhattan) distance separating two agents in a basic grid.
	A route followed by an agent in a basic grid corresponds to a path in the grid (\ie a sequence of edges
	$e_1,e_2,e_3,e_4,\ldots$) that are consecutively traversed by the agent until rendezvous is done. For any
	integer $k$, we define the {\em reverse path} to the path $e_1,\dots,e_k$ as the path $e_k,e_{k-1},\dots,
	e_1=\overline{e_1,\dots,e_{k-1},e_k}$. We denote by $C(p)$ the number of edge traversals performed by an agent
	during the execution of a procedure $p$.

	Consider two distinct nodes $u$ and $v$. We define a specific path from $u$ to $v$, denoted $\pa(u, v)$, as
	follows. If there exists a unique shortest path from $u$ to $v$, this {shortest path} is $\pa(u, v)$.
	Otherwise, consider the smallest rectangle $R_{(u,v)}$ such that $u$ and $v$ are two of its corners. $\pa(u,
	v)$ is the unique path among the shortest path from $u$ to $v$ that traverses all the edges on the northern
	side of $R_{(u,v)}$. Note that $\pa(u, v) = \overline{\pa(v, u)}$. % 
\ifTR

	% \begin{definition}[reverse path]
	%  Following the reverse path of a path $\pi$, an agent follows each edge
	%  of $\pi$ in reverse order and opposite direction. The reverse path of
	%  the path $u_1, \ldots, u_i$ is the path $u_i, \ldots, u_1$, and we note it
	%  $\overline{u_1, \ldots, u_i}$. When an agent follows the path $\overline{u_1,
	%  \ldots, u_i}$, we also say that it backtracks the path $u_1, \ldots, u_i$.
	% \end{definition}

	An illustration of $\pa(u, v)$ is given in Figure~\ref{path}.

	%At least one of them is a Northern corner. In this
	% rectangle, there is a path from $u$ to $v$ which follows the Northern side,
	% and the Western or the Eastern side: this path is $\pa(u,
	% v)$.

	\begin{figure}[!h]
		\hspace*{0mm}\vfill
		\begin{center}
			\includegraphics[scale=0.3]{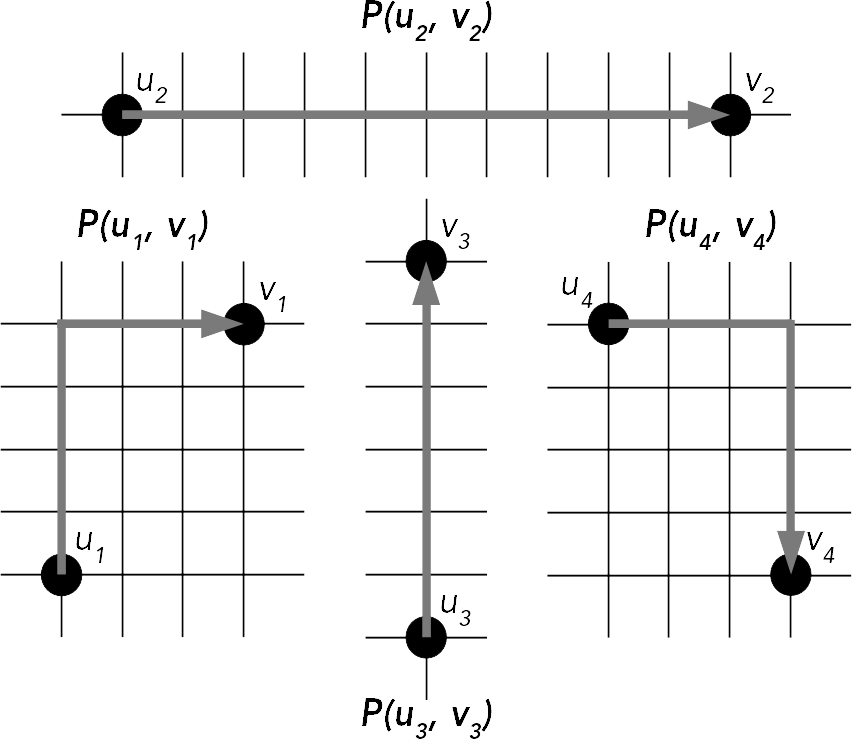}
		\end{center}
		\vfill\hspace*{0mm}
		\caption{Some different cases for $\pa(u, v)$}
		\label{path}
	\end{figure}
\fi

\section{Idea of the algorithm}
\label{sec:idea}

In this section we give the high level idea of our solution: more detailed explanations are given in section~\ref{sec:principalAlgorithm}.

%	This section is dedicated to the high level idea behind the algorithm.

\ifTR
  \subsection{Informal Description in a Nutshell}
  \label{int:nutshell} 
\else
  \ifSODA

   % \vspace*{-6pt}

    \paragraph{Informal Description in a Nutshell.}

%     \label{int:nutshell}
  \fi
\fi

		We aim at achieving rendezvous of two asynchronous mobile agents in an infinite grid and in a deterministic way. It is well known that solving rendezvous deterministically is impossible in some symmetric graphs (like a basic grid) unless both agents are given distinct identifiers called labels. We use them to break the symmetry, \ie in our context, to make the agents follow different routes. The idea is to make each agent ``read'' its label binary representation, {one bit at a time} from the most to the least significant bits, and for each bit it reads, follow a route depending on the read bit. Our algorithm ensures rendezvous during some of the periods when they follow different routes \ie  when the two agents process two different bits.

		Furthermore, to design the routes that both agents will follow, our approach would require to know an upper bound on two parameters, namely the initial distance between the agents and the length (of the binary representation) of the shortest label. As we suppose that the agents have no knowledge of these parameters, they both perform successive ``assumptions'', in the sequel called \emph{phases}, in order to find out such an upper bound. Roughly speaking, each agent attempts to estimate such an upper bound by successive tests, and for each of these tests, acts as if the upper bound estimation was correct. Both agents first perform Phase $0$. When Phase~$i$ does not lead to rendezvous, they perform Phase~$i+1$, and so on. More precisely, within Phase~$i$, the route of each agent is built in such a way that it ensures rendezvous if $2 ^ i$ is a good upper bound on the parameters of the problem. Hence, in our approach two requirements are needed: both agents are assumed $(1)$ to process two different bits (\ie $0$ and $1$) almost concurrently and $(2)$ to perform Phase~$i=\alpha$ almost at the same time---where $\alpha$ is the smallest integer such that the two aforementioned parameters are upper bounded by $2^{\alpha}$.

		However, to meet these requirements, we have to face two major issues.  First, since the adversary can vary both agent speeds, the idea described above does not prevent the adversary from making the agents always process the same type of bit at the same time. {Moreover}, the route cost depends on the phase number, and thus, if an agent were performing some Phase~$i$ with $i$ exponential in the initial distance and in the length of the binary representation of the smallest label, then our algorithm would not be polynomial. To tackle these two issues, we use a mechanism that prevents the adversary from making an agent execute the algorithm arbitrarily faster than the other without meeting. Each of both these issues is {circumvented} via a specific ``synchronization mechanism''. Roughly speaking, the first one makes the agents read and process the bits of the binary representation of their labels at {nearly} the same speed, while the second ensures that they start Phase~$\alpha$ at almost the same time. This is particularly where our feat of strength is: orchestrating in a subtle manner these synchronizations in a fully asynchronous context while ensuring a polynomial cost. Now that we have described the very high level idea of our algorithm, let us give more details. 

\ifTR

\subsection{Under the hood}
\label{hood}
\else
  \ifSODA

   % \vspace*{-6pt}

    \paragraph{Under the Hood.}

%     \label{int:nutshell}
  \fi
\fi

The approach described above allows us to solve rendezvous when there exists an index for which the binary representations of both labels differ. However, this is not always the case especially when a binary representation is a prefix of the other one (e.g., $100$ and $1000$). Hence, instead of considering its own label, each agent will consider a transformed label: The transformation borrowed from~\cite{DessmarkFKP06} will guarantee the existence of the desired difference over the new labels. In the rest of this description, we assume for convenience that the initial Manhattan distance $D$ separating the agents is at least the length of the shortest binary representation of the two transformed labels (the complementary case adds an unnecessary level of complexity to understand the intuition).

As mentioned previously, our solution ({cf. Algorithm~\ref{AsynchronousUnknownDistance} in Section~\ref{sec:principalAlgorithm}}) works in phases numbered $0,1,2,3,4,\ldots$ During Phase~$i$ ({cf. Procedure~$\ite$ called at line~\ref{calliteration} in Algorithm~\ref{AsynchronousUnknownDistance}}), the agent supposes that the initial distance $D$ is at most $2^{i}$ and processes one by one the first $2^{i}$ bits of its transformed label: In the case where $2^{i}$ is greater than the binary representation of its transformed label, the agent will consider that each of the last ``missing'' bits is $0$. When processing a bit, the agent executes a particular route which depends on the bit value and the phase number. The route related to bit $0$ ({relying in particular on Procedure~$\fl$ called at line~\ref{callFL} in Algorithm~\ref{iteration}}) and the route related to bit $1$ ({relying in particular on Procedure~$\bo$ called at line~\ref{callBO} in Algorithm~\ref{iteration}}) are obviously different and designed in such a way that if both these routes are executed almost simultaneously by two agents within a phase corresponding to a correct upper bound, then rendezvous occurs by the time any of them has been completed.

%{Actually, for any non-negative integer $i$, each bit process within Phase~$i$ consists in some polynomial in $i$ number of steps during which each agent executes a particular route which depends on the bit value, the phase number $i$ and the step number. Hence, the two different routes related to bit processing consist in a sequence of subroutes (Procedure~$\fl$ for bit $0$ and Procedure~$\bo$ for bit $1$) which are parametrized by the step number. The idea behind this step number is the following. It can be viewed as the current supposition made by the agent processing bit $1$ on the initial node of the other agent. More precisely, in order to ensure rendezvous by the time any of them has been completed, the subroute related to bit $0$ and the subroute related to bit $1$ require, not only to be executed almost simultaneously by two agents within a phase corresponding to a correct upper bound but also within a step corresponding to the initial node of the agent processing bit $0$.}

In the light of this, if we denote by $\alpha$ the smallest integer such that $2^{\alpha}\geq D$, it turns out that an ideal situation would be that the agents concurrently start phase $\alpha$ {and process the bits at quite the same rate} within this phase. Indeed, we would then obtain the occurrence of rendezvous by the time the agents complete the process of the {$\lambda$}th bit of their transformed label in phase $\alpha$, where {$\lambda$} is the smallest index for which the binary representations of their transformed labels differ. However, getting such an ideal situation in presence of a fully asynchronous adversary appears to be really challenging. This is where the two synchronization mechanisms briefly mentioned above come into the picture.

{If the agents start Phase~$\alpha$ approximately at the same time, the first synchronization mechanism ({cf. Procedure $\rd$ called at line~\ref{BitSynchronization} in Algorithm~\ref{iteration}}) permits to force the adversary to make the agents process their respective bits at similar speed within Phase~$\alpha$, as otherwise rendezvous would occur prematurely during this phase before the process by any agent of the {$\lambda$}th bit. This constraint is imposed on {the adversary by dividing each bit process into some predefined steps and by ensuring} that after each step $s$ of the $k$th bit process, for any $k\leq 2^{\alpha}$, an agent follows a specific route that forces the other agent to complete the step $s$ of its $k$th bit process.} {This route, on which the first synchronization is based, is constructed by relying on a simple principle that enables an agent to ``push'' the other. The principle is as follows: if an agent performs a given route $X$ included in a given area $\mathcal{S}$ of the basic grid, then the other agent can force it to finish route $X$ by covering $\mathcal{S}$ as many times as there are edge traversals in $X$. More precisely, each covering of $\mathcal{S}$  allows to traverse all the edges of $X$ at least once: so, in each covering the agent executing $X$ must complete at least one edge traversal or rendezvous occurs.} Hence, one of the major difficulties we have to face lies in the setting up of the second synchronization mechanism guaranteeing that the agents start Phase~$\alpha$ around the same time. At first glance, it might be tempting to use an analogous principle to the one used for dealing with the first synchronization. Indeed, if an agent $a_1$ follows a route covering $r$ times an area $\mathcal{Y}$ of the grid, such that $\mathcal{Y}$ is where the first $\alpha-1$ phases of an agent $a_2$ take place and $r$ is the maximal number of edge traversals an agent can make during these phases, then agent $a_1$ pushes agent $a_2$ to complete its first $\alpha-1$ phases and to start Phase~$\alpha$. Nevertheless, a strict application of this principle to the case of the second synchronization directly leads to an algorithm having a cost that is superpolynomial in $D$ and the length of the smallest label, due to a cumulative effect that does not appear for the case of the first synchronization. As a consequence, to force an agent to start its Phase~$\alpha$, the second synchronization mechanism does not depend on the kind of route described above, but on a much more complicated route that permits an agent to ``{\em push}'' the second one. This works by considering the ``pattern" that is drawn on the grid by the second agent rather than just the number of edges that are traversed ({cf. Procedure $\p$ called at line~\ref{MainFireworks'Call} in Algorithm~\ref{iteration}}). This is the most tricky part of our algorithm, one of the main idea of which relies in particular on the fact that some routes made of an arbitrarily large sequence of edge traversals can be pushed at a relative low cost by some other routes that are of comparatively small length, provided they are judiciously chosen. Let us illustrate this point through the following example. Consider an agent $a_1$ following from a node $v_1$ an arbitrarily large sequence of $X_i$, in which each $X_i$ corresponds either to $A\overline{A}$ or $B\overline{B}$ where $A$ and $B$ are any routes ($\overline{A}$ and $\overline{B}$ corresponding to their respective backtrack \ie the sequence of edge traversals followed in the reverse order). An agent $a_2$ starting from an initial node $v_2$ located at a distance at most $d$ from $v_1$ can force agent $a_1$ to finish its sequence of $X_i$ (or otherwise rendezvous occurs), regardless of the number of $X_i$, simply by executing $A\overline{A}B\overline{B}$ from each node at distance at most $d$ from $v_2$. To support this claim, let us suppose by contradiction that it does not hold. At some point, agent $a_2$ necessarily follows $A\overline{A}B\overline{B}$ from $v_1$. However, note that if either agent starts following $A\overline{A}$ (resp. $B\overline{B}$) from node $v_1$ while the other is following $A\overline{A}$ (resp. $B\overline{B}$) from node $v_1$, then the agents meet. Indeed, this implies that the more ahead agent eventually follows $\overline{A}$ (resp. $\overline{B}$) from a node $v_3$ to $v_1$ while the other is following $A$ (resp. $B$) from $v_1$ to $v_3$, which leads to rendezvous. Hence, when agent $a_2$ starts following $B\overline{B}$ from node $v_1$, agent $a_1$ is following $A\overline{A}$, and is not in $v_1$, so that it has at least started the first edge traversal of $A\overline{A}$. This means that when agent $a_2$ finishes following $A\overline{A}$ from $v_1$, $a_1$ is following $A\overline{A}$, which implies, using the same arguments as before, that they meet before either of them completes this route. Hence, in this example, agent $a_2$ can force $a_1$ to complete an arbitrarily large sequence of edge traversals with a single and simple route. {Actually, our second synchronization mechanism implements this idea (this point is refined in Section~\ref{sec:principalAlgorithm})}. %Roughly speaking, to make them pushed by the second synchronization mechanism at low cost, the phases are designed in such a way that some parts of them can be pushed at low cost in a similar manner as the route followed by agent $a_1$ in the above example.
{This was the most complicated thing to set up, as each part of route in every phase had to be orchestrated very carefully to permit, in the end, this low cost synchronization while still ensuring rendezvous. However, it is through this original and novel way of moving that we finally get the polynomial cost.}

\section{Basic patterns} \label{sec:bas}

	In this section we define some sequences of moving instructions, \ie patterns of moves, that will serve in turn as building blocks in the construction of our rendezvous algorithm. The main roles of these patterns are given in the next section when presenting our general solution.

	%Each of patterns $\di$, $\fl$ and $\bo$ is exactly made of two periods such that the second one is a backtrack of the first one.
	%Hence, when finishing any of those three patterns, an agent is always back at the node where it started it.

	% \begin{remark} \label{bas:rem}
	%  Consider the reverse paths of $\di$, $\fl$ or $\bo$. If the path
	%  followed during one of those patterns first periods is denoted
	%  $\sigma$, the reverse path of the whole pattern is $\overline{\sigma,
	%  \overline{\sigma}}$. By definition of the reverse path, it can also
	%  be written $\overline{\overline{\sigma}}, \overline{\sigma} = \sigma,
	%  \overline{\sigma}$. This means that the reverse paths of $\di$, $\fl$
	%  and $\bo$ are respectively $\di$, $\fl$ and $\bo$.
	% \end{remark}

	\subsection{Pattern $\di$}
	
	\begin{figure}[!h]
			\hspace*{0mm}\vfill
			\begin{center}
				\includegraphics[scale=0.26]{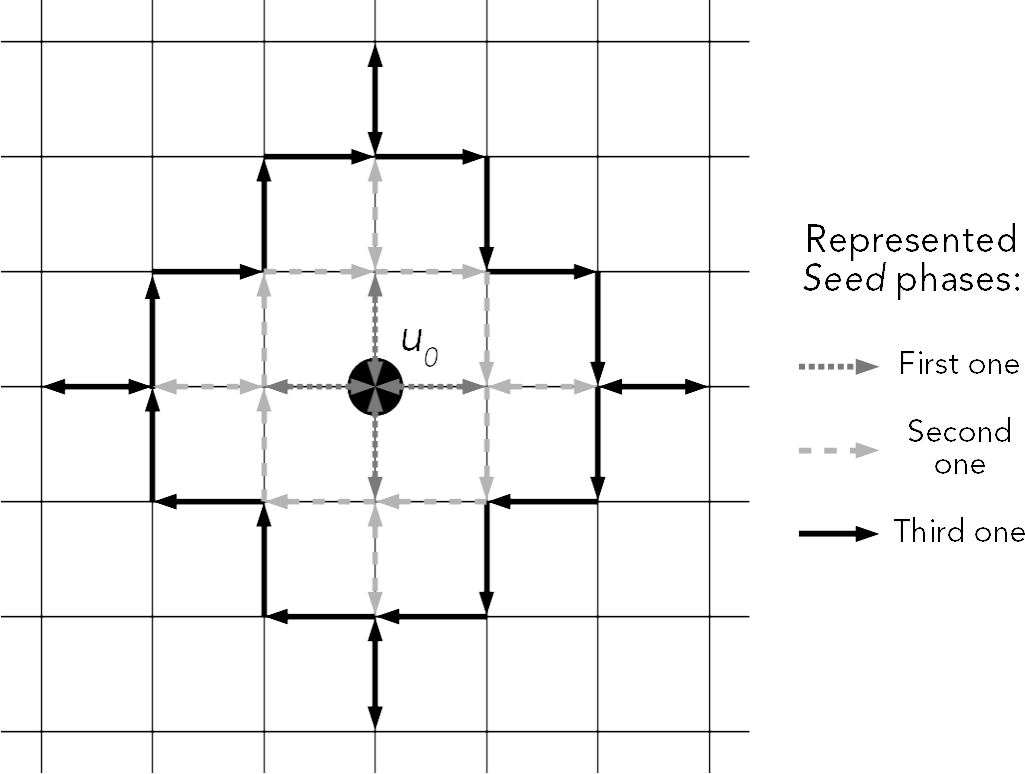}
			\end{center}
			\vfill\hspace*{0mm}
			\caption{{An illustration of the movements executed by an agent during the first period of $\di(3)$ from a node $u_0$. An arrow from a node $x$ to a node $y$ represents an edge traversal from $x$ to $y$. Depending on the shape of the arrow, the represented movement is performed in a different phase.}}
			\label{fig:snail}
		\end{figure}
		Pattern $\di$ is involved as a subpattern in the design of all the other patterns presented in this section.
		
		The description of Pattern~$\di$ is given in Algorithm~\ref{bas:di}. It is made of two periods. For a given non-negative integer $x$, the first period of Pattern $\di(x)$ corresponds to the execution of $x$ phases, while the second period is a complete backtrack of the path travelled during the first period. Pattern $\di$ is designed in such a way that it offers some properties that are shown in {Section}~\ref{sss:di} and that are necessary to conduct the proof of correctness. One of the main purpose of this pattern is the following: starting from a node $v$, Pattern $\di(x)$ allows to visit all nodes of the grid at distance at most $x$ from $v$ and to traverse all edges of the grid linking two nodes at distance at most $x$ from $v$ (informally, the procedure permits to cover an area of radius $x$). {An illustration of Pattern~$\di$ is given in Figure~\ref{fig:snail}}.

		\begin{algorithm}[!h]
			\caption{Pattern $\di(\dip)$} \label{bas:di}
			\begin{algorithmic}[1]
				\STATE /* First period */
				\FOR {$i \leftarrow 1$; $i \leq \dip$; $i \leftarrow i + 1$}
					\STATE /* Phase $i$ */
					\STATE Perform $(N (SE)^i (WS)^i (NW)^i (EN)^i)$
				\ENDFOR
				\STATE /* Second period */
				\STATE $L \leftarrow$ the path followed by the agent during the first period
				\STATE Backtrack by following the reverse path $\overline{L}$
			\end{algorithmic}
		\end{algorithm}

	\subsection{Pattern $\rd$}
	
	Following the high level description of our solution (Section~\ref{sec:idea}), Pattern $\rd$ is the basic
	primitive procedure that implements the first synchronization mechanism (between two consecutive steps of a bit
	process). An agent $a_1$ executing pattern $\rd(\dip, n)$ from a node $u$ processes $n$ times pattern
	$\di(\dip)$ from node $u$. All along this execution, $a_1$ stays at distance at most $\dip$ from $u$. {Moreover},
	once the execution is over, the agent is back at $u$.

		The description of pattern $\rd$ is given in Algorithm~\ref{bas:rd}.

		\begin{algorithm}
			\caption{Pattern $\rd(\dip, n)$} \label{bas:rd}
			\begin{algorithmic}
				\STATE Execute $n$ times Pattern $\di(\dip)$
			\end{algorithmic}
		\end{algorithm}

	\subsection{Pattern $\fl$}
	
                 According to Section~\ref{sec:idea}, Pattern $\fl$ is used in particular to design the specific route that an agent follows when processing bit $0$. The description of Pattern $\fl$ is given in Algorithm~\ref{bas:fl}. It is made of two periods, the second of which is a backtrack of the first one. Pattern $\fl$ offers several properties that are proved in {Section}~\ref{sss:fl} and used in the proof of correctness. Note that, Pattern $\fl(\dip, \flp)$ executed from a node $u$ for any two integers $\dip$ and $\flp$ allows, {in particular}, an agent to perform Pattern $\di(\dip)$ from each node at distance at most $\flp$ from $u$. {An illustration of Pattern~$\fl$ is given in Figure~\ref{fig:flower}}.

		\begin{algorithm}
			\caption{Pattern $\fl(\dip, \flp)$} \label{bas:fl}
			\begin{algorithmic}[1]
				\STATE /* First period */
				\STATE Let $u$ be the current node
				\FOR {$i \leftarrow 1$; $i \leq \dip + \flp$; $i \leftarrow i + 1$}
					\FOR{$j \leftarrow 0$; $j \leq i$; $j \leftarrow j + 1$}
						\FOR{each node $v$ at distance $j$ from $u$ ordered in the
						clockwise direction from the North}
							\STATE Follow $\pa(u, v)$
							\STATE Execute $\di(i - j)$
							\STATE Follow $\pa(v, u)$
						\ENDFOR
					\ENDFOR
				\ENDFOR
				\STATE /* Second period */
				\STATE $L \leftarrow$ the path followed by the agent during the first period
				\STATE Backtrack by following the reverse path $\overline{L}$
			\end{algorithmic}
		\end{algorithm}
                 
%		{Figure~\ref{fig:flower} displays a part of the route followed by an agent executing pattern $\fl(2, 3)$ from node $u_0$. Inside the perimeter of pattern $\fl$ which is represented by the main square, there are all the nodes at distance at most 3. As this pattern consists of many patterns $\di$, only a few of them are represented. In particular, inside $\fl(2, 3)$, there are patterns $\di(2)$ from every node at distance at most 3 from $u_0$, but there also are patterns $\di$ from nodes at a distance up to 5 from $u_0$, and all patterns $\di$ called by $\fl(2, 3)$ are not asigned an equal to 2 parameter. The represented patterns $\di$ are among the patterns $\di(2)$ from nodes at distance at most 3 from $u_0$ and are executed from $u_1$, $u_2$ and $u_3$. Doted squares represent the perimeter of patterns $\di$: inside them, there are all the nodes at distance at most 2 from the initial node of the pattern $\di$ (which can be $u_1$, $u_2$ or $u_3$). When going to execute $\di(2)$ from node $u_1$ (resp. $u_2$ or $u_3$), an agent follows $\pa(u_0, u_1)$ (resp. $\pa(u_0, u_2)$ or $\pa(u_0, u_3)$). After executing them, the agent follows the reverse paths which are $\pa(u_1, u_0)$, $\pa(u_2, u_0)$ and $\pa(u_3, u_0)$. These different $\pa$ are represented by the arrows between $u_0$ and the three initial nodes of the patterns $\di$.}

		\begin{figure}[!h]
			\hspace*{0mm}\vfill
			\begin{center}
				\includegraphics[scale=0.3]{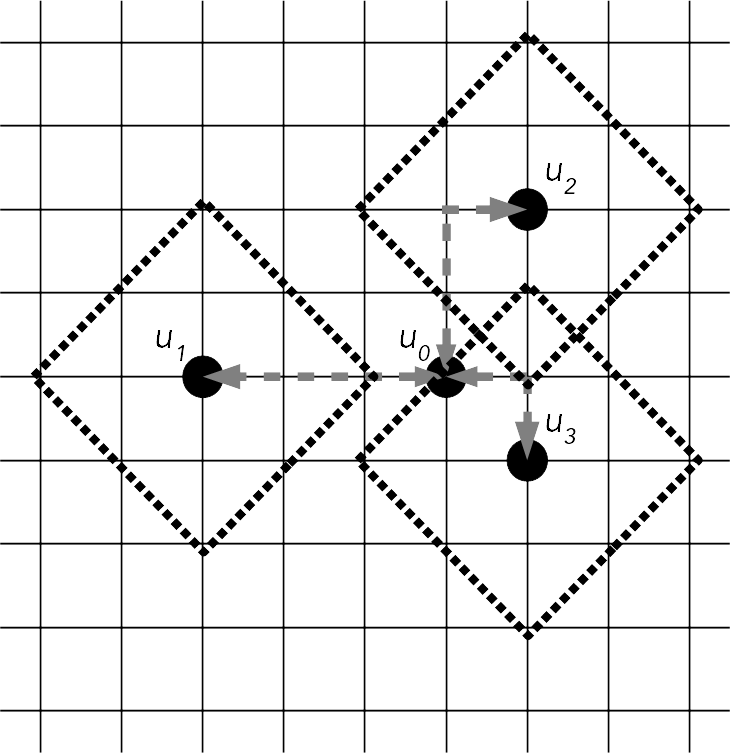}
			\end{center}
			\vfill\hspace*{0mm}
			\caption{{Illustration of a part of the route followed by an agent executing Pattern $\fl(2, 3)$ from a node $u_0$. When executing this pattern the agent has to execute many patterns $\di$ interleaved with executions of paths $\pa$ from all nodes at distance at most $3$ from $u_0$. Some of these patterns and paths are depicted in the figure. It is particularly the case of the dotted square centered at $u_1$ (resp. $u_2$ and $u_3$) that delimits the set of nodes that are visited when executing a pattern $\di(2)$ from node $u_1$ (resp. $u_2$ and $u_3$). Before executing $\di(2)$ from node $u_1$ (resp. $u_2$ or $u_3$), the agent follows $\pa(u_0, u_1)$ (resp. $\pa(u_0, u_2)$ or $\pa(u_0, u_3)$), and after executing $\di(2)$ from node $u_1$ (resp. $u_2$ or $u_3$), the agent follows the path $\pa(u_1, u_0)$ (resp. $\pa(u_2, u_0)$ or $\pa(u_3, u_0)$). These different paths $\pa$ are represented by arrows.}}
			\label{fig:flower}
		\end{figure}
\newpage

	\subsection{Pattern $\bo$}
	
		According to Section~\ref{sec:idea}, Pattern~$\bo$ is used in particular to design the specific route that an agent follows when processing bit $1$. The description of Pattern $\bo$ is given in Algorithm~\ref{bas:bo}. As for Patterns $\di$ and $\fl$, the pattern is made of two periods, the second of which corresponds to a backtrack of the first one. Properties related to this pattern are given in {Section}~\ref{sss:bo}. Note that, Pattern $\bo(\dip, \flp, \bop, \bops)$ executed from a node $u$ for any integers $\dip$, $\flp$, $\bop$ and $\bops$ allows an agent to perform {Patterns~$\fl(\dip, \flp)$ and~$\di(\dip)$} from each node at distance at most $\bop$ from $u$. %{Parameter~$\bops$ corresponds to the step number, introduced in Subsection~\ref{hood}, which can be viewed as a supposition on the initial node of the other agent. Its value determines in particular the first node (among the $2\bop (\bop + 1) + 1$ at distance at most $\bop$ from $u$) from which Patterns~$\fl(\dip, \flp)$ and~$\di(\dip)$ are executed within Pattern $\bo(\dip, \flp, \bop, \bops)$.}
{Parameter $h$ is an integer input that indicates in which order the agent has to visit each node at distance at most $\bop$ from $u$ (to execute Patterns~$\fl(\dip, \flp)$ and~$\di(\dip)$ from each of these nodes). Playing on this order is used for technical reasons that are detailed in the proof of Theorem~\ref{lemma6Cases}.} {An illustration of Pattern~$\bo$ is given in Figure~\ref{fig:bouquet}}.

		\begin{algorithm}
			\caption{Pattern $\bo(\dip, \flp, \bop, \bops)$} \label{bas:bo}
			\begin{algorithmic}[1]
				\STATE /* First period */
				\STATE Let $u$ be the current node
				\STATE Let $U$ be the list of nodes at distance at most $\bop$ from $u$ ordered in the order of the first visit when applying $\di(\bop)$ from node $u$
				\FOR {$i \leftarrow 0$; $i \leq 2\bop (\bop + 1)$; $i \leftarrow i + 1$}
					\STATE Let $v$ be the node with index $\bops + i \pmod{2\bop (\bop + 1) + 1}$ in $U$
					\STATE Follow $\pa(u,v)$
					\STATE Execute $\di(\dip)$
					\STATE Execute $\fl(\dip, \flp)$
					\STATE Follow $\pa(v,u)$
				\ENDFOR
				\STATE /* Second period */
				\STATE $L \leftarrow$ the path followed by the agent during the first period
				\STATE Backtrack by following the reverse path $\overline{L}$
			\end{algorithmic}
		\end{algorithm}

		\begin{figure}[!h]
			\hspace*{0mm}\vfill
			\begin{center}
				\includegraphics[scale=0.3]{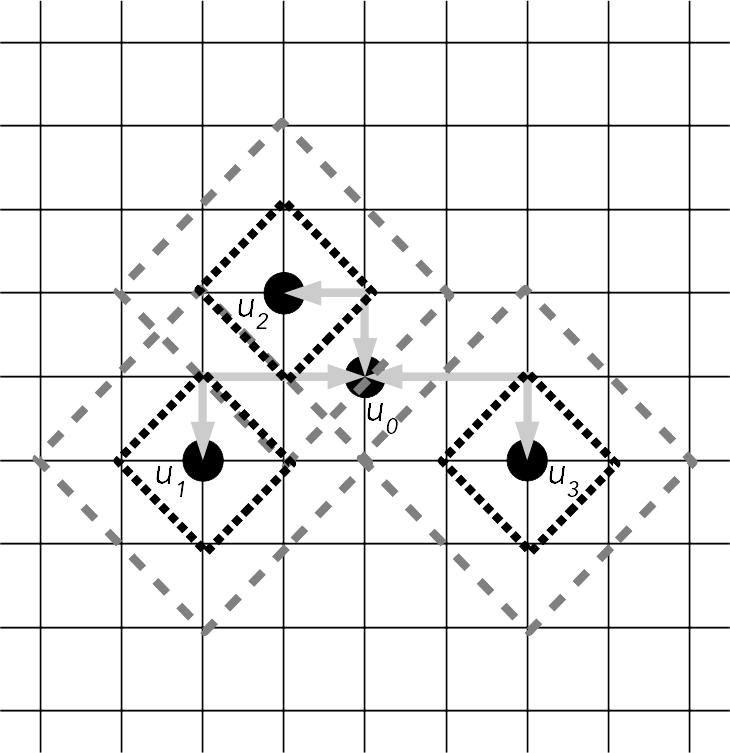}
			\end{center}
			\vfill\hspace*{0mm}
			\caption{ {Illustration of a part of the route followed by an agent executing Pattern $\bo(1, 2, 3, 0)$ from a node $u_0$. When executing this pattern the agent has to execute paths $\pa$ as well as patterns $\di$ and $\fl$ from all nodes at distance at most $3$ from $u_0$ and in particular from nodes $u_1$, $u_2$ and $u_3$. To go to these nodes from $u_0$, the agent respectively follows $\pa(u_0, u_1)$, $\pa(u_0, u_2)$ and $\pa(u_0, u_3)$. Once in node $u_1$ (resp. $u_2$ and $u_3$) the agent executes $\di(1)$, which is represented by the smallest dotted square centered at $u_1$ (resp. $u_2$ and $u_3$) and then executes $\fl(1, 2)$, which is represented by the largest dotted square centered at $u_1$ (resp. $u_2$ and $u_3$), followed by $\pa(u_1, u_0)$ (resp. $\pa(u_2, u_0)$ and $\pa(u_3, u_0)$). All paths $\pa$ are represented by arrows.}}
			\label{fig:bouquet}
		\end{figure}

\section{Main Algorithm} \label{sec:principalAlgorithm}

{In this section, we give the formal description of our solution allowing to solve rendezvous in a basic grid. We also give the main objectives of the involved subroutines and how they work at a high level. The main algorithm that solves the rendezvous in a basic grid is Algorithm~\aud{} (shown in Algorithm~\ref{AsynchronousUnknownDistance})}. As mentioned in Subsection~\ref{hood}, we use the label of an agent only when it has been transformed. Let us describe this transformation that is borrowed from~\cite{DessmarkFKP06}. Let $(b_0 b_1 \ldots b_{n-1})$ be the binary representation of the label of an agent. We define its transformed label as the binary sequence $(b_0 b_0 b_1 b_1 \ldots b_{n-1} b_{n-1} 0 1)$. This transformation permits to obtain the feature that is highlighted by the following remark.

	\begin{remark} \label{alg:rem}
		Given two distinct labels $l_a$ and $l_b$, their transformed labels are never prefixes of each other. In other words, there exists an index {$\lambda$} such that the {$\lambda$}th bit of the transformed label of $l_a$ is different from the {$\lambda$}th bit of the transformed label of $l_b$.
	\end{remark}

As explained in Section~\ref{sec:idea}, we need such a feature because our solution requires that at some point both agents follow different routes by processing different bit values.

\begin{algorithm}[H]
		\caption{\aud} \label{AsynchronousUnknownDistance}
		\begin{footnotesize}
		\begin{algorithmic}[1]
%			\STATE Let $Label$ be the label of the agent represented as an array of bits indexed from 0, and $n$ its length
%			\STATE Let $TransformedLabel$ be an array of length $2n + 2$ indexed from 0
%			\FOR {each bit $b_i$ of $Label$}
%				\STATE $TransformedLabel[2i] = b_i$
%				\STATE $TransformedLabel[2i + 1] = b_i$
%			\ENDFOR
%			\STATE $TransformedLabel[2n] = 0$
%			\STATE $TransformedLabel[2n + 1] = 1$
			\STATE $d \leftarrow 1$
			\WHILE {agents have not met yet} \label{main loop}
				\STATE Execute $\ite(d)$ \label{calliteration}
				\STATE $d \leftarrow 2d$ \label{d doubled}
			\ENDWHILE
		\end{algorithmic}
		\end{footnotesize}
	\end{algorithm}

{Algorithm \aud{} makes use of a subroutine, \ie Procedure~$\ite$. When an agent executes this procedure with a parameter $\alpha$ that is a ``good'' assumption \ie that upperbounds the initial distance $D$ and the value {$\lambda$} of the smallest bit position for which both transformed labels differ, we have the guarantee that rendezvous occurs by the end of this execution. In the rest of this section, we assume that $\alpha$ is the smallest good assumption that upperbounds $D$ and {$\lambda$}.} 

%It makes use, for technical reasons, of the sequence $r$ that is defined below.
%	\begin{gather*}
%		\forall \text{ power of two } i,\text{ } \rho(i) = 2i^4 \text{ and } r(i) = \rho(i) + 3i
%	\end{gather*}
	
{The code of Procedure~$\ite$ is given in Algorithm~\ref{iteration}. It} can be divided into two parts. The first part consists of the execution of Procedure~$\p$ (line 1 of Algorithm~\ref{iteration}) and corresponds to the second synchronization mechanism mentioned in Section~\ref{sec:idea}. The main feature of this procedure is the following: when the earlier agent finishes the execution of $\p(\alpha)$ within the execution of $\ite(\alpha)$, we have the guarantee that the later agent has at least started to execute $\ite$ with parameter $\alpha$ (actually, as explained below, we have even the guarantee that most of $\p(\alpha)$ has been executed by the later agent). Procedure~$\p$ is presented below. The second part of Procedure~$\ite$ (cf. lines $2-19$ of Algorithm~\ref{iteration}) consists in processing the bits of the transformed label one by one. More precisely when processing a given bit in a call to Procedure $\ite(d)$, the agent acts in steps $0,1,\ldots,2d(d+1)$: After each of these steps, the agent executes Pattern $\rd$ whose role is described below. In each of these steps, the agent executes $\fl$ (resp. $\bo$) if the bit it is processing is $0$ (resp. $1$). These patterns of moves (cf. Algorithms~\ref{bas:fl} and~\ref{bas:bo} in Section~\ref{sec:bas}) are made in such a way that rendezvous occurs by the time any agent finishes the process of its {$\lambda$}th bit in $\ite(\alpha)$ if we have the following synchronization property. Each time any of the agents starts executing a step {$j$} during the process of its {$i$}th bit in $\ite(\alpha)$, the other agent has finished the execution of either step {$j-1$} in the {$i$}th bit process of $\ite(\alpha)$ if {$j>0$}, or the last step of the {$(i-1)$}th bit process of $\ite(\alpha)$ if {$j=0$ and $i>0$}.
% (Procedure~$\p(\alpha)$ ensures that when any agent starts executing the first step during the process of the first bit the other agent has nearly finished the execution of $\p(\alpha)$ and thus started the execution of the first step within the first bit process.)  
To obtain such a synchronization, an agent executes what we called the first synchronization mechanism in the previous section (cf. line $15$ in Algorithm~\ref{iteration}) after each step of a bit process. Actually, this mechanism relies on procedure $\rd$, the code of which is given in Algorithm~\ref{bas:di}. Note that the total number of steps, and thus of executions of $\rd$, in $\ite(\alpha)$ is $2\alpha^2(\alpha+1)+\alpha$. For every {$0\leq k \leq 2\alpha^2(\alpha+1)+\alpha$}, the {$k$}th execution of $\rd$ in $\ite(\alpha)$ by an agent permits to force the other agent to finish the execution of its {$k$}th step in $\ite(\alpha)$ by repeating a pattern $\di$ (its main purpose is described just above its code given by Algorithm~\ref{bas:rd}): {with} the appropriate parameters, this pattern $\di$ covers any pattern ($\fl$ or $\bo$) made in the {$k$}th step of $\ite(\alpha)$ and the number of times it is repeated is at least the maximal number of edge traversals we can make in the {$k$}th step of $\ite(\alpha)$.

% The codes of $\p$, and $\pp$ are given in Algorithm~\ref{Pprocedure}, and
% Algorithm~\ref{patternPP}, respectively. Algorithm $\p$ 
% corresponds to the second synchronization mechanism mentioned in Section~\ref{sec:idea}, while Algorithm $\pp$ 
% is a subroutine of the former one
% allowing to push an agent under some conditions (or otherwise rendezvous occurs).

% To formally describe Algorithms~\ref{iteration} and~\ref{Pprocedure}, we need to define two sequences 
% that will be used in the instructions of both these algorithms:
%	\begin{gather*}
%		\rho(1) = 1 \\
%		\forall \text{ power of two } i,\text{ } r(i) = \rho(i) + 3i \\
%		\forall \text{ power of two } i\geq 2,\text{ }\rho(i) = r(\frac {i} {2}) + \frac {3i} {2} (\frac {i} {2} (i (\frac {i} {2} + 1) + 1) + 1)
%	\end{gather*}

{Algorithm~\ref{Pprocedure} gives the code of Procedure~$\p$. Procedure~$\p$ is made of two parts: the executions of Procedure~$\pp$ (lines~$1-3$ of Algorithm~\ref{Pprocedure}), and the calls to the patterns $\bo$ and $\rd$ (lines~$4-5$ of Algorithm~\ref{Pprocedure}). When $\p$ is executed with parameter $\alpha$ (which is a good assumption), the first part ensures that the later agent has at least completed every execution of $\ite$ with a parameter that is smaller than $\alpha$, while the second part ensures that the later agent has completed almost the entire execution of $\p(\alpha)$ (more precisely, when the earlier agent finishes the second part, we have the guarantee that {it remains for the later agent to execute at most the last line before completing its own execution of $\p(\alpha)$).}}

%\vspace{-0.6cm}

%\vspace{-0.6cm}

	\begin{algorithm}[H]
		\caption{$\ite(d)$} \label{iteration}
		\begin{footnotesize}
		\begin{algorithmic}[1]
			\STATE Execute $\p(d)$ \label{MainFireworks'Call}
			\STATE {$radius \leftarrow 2d^4+3d$}
			\STATE $i \leftarrow 1$
			\WHILE {$i \leq d$} \label{secondLoop}
				\STATE $j \leftarrow 0$
				\WHILE {$j \leq 2d (d + 1)$} \label{firstLoop}
					\STATE // Begin of step $j$
%					\IF{$TransformedLabel[i] = 0$ or $i$ is at least the length of $TransformedLabel$}
					\IF{the length of the transformed label is strictly greater than $i$, or its $i$th bit is $0$} \label{readingTest}
						\STATE Execute $\fl(radius, d)$ \label{callFL}
					\ELSE
						\STATE Execute $\bo(radius, d, d, j)$ \label{callBO}
					\ENDIF
					\STATE // End of step $j$
					\STATE $radius \leftarrow radius + 3d$
					\STATE Execute $\rd(radius, C(\bo(radius - 3d, d, d, j)))$ \label{BitSynchronization}
					\STATE $j \leftarrow j + 1$
				\ENDWHILE
				\STATE $i \leftarrow i + 1$
			\ENDWHILE
		\end{algorithmic}
		\end{footnotesize}
	\end{algorithm}

%\vspace{-0.3cm}

{To give further details on Procedure~$\p$, let us first describe Procedure~$\pp$ (its code is given in Algorithm~\ref{patternPP}). When the earlier agent completes the execution of $\pp(2i, d)$ with $i$ some power of two, assuming that the later agent had already completed $\ite(i)$, we have the guarantee that the later agent has completed its execution of $\ite(2i)$. To ensure this, we regard the execution of $\ite(2i)$ as a sequence of calls to basic patterns (namely $\rd$, $\fl$ and $\bo$), which is formally defined in Definition~\ref{def:decompo}. This sequence is what we meant when talking about ``the pattern drawn on the grid'' in Subsection~\ref{hood}. {The sequence of calls to basic patterns of the earlier agent in $\ite(2i)$ is quite similar to the one of the later agent: they have the same length and the $s$th pattern of one sequence is $\rd$ iff the $s$th pattern of the other sequence is $\rd$. In fact, the only difference, due to distinct transformed labels, is that if the $s$th pattern of one sequence is $\fl$ (resp. $\bo$), the $s$th pattern of the other sequence may be either $\fl$ or $\bo$.}

{For each basic pattern $p_s$ in its sequence, the earlier agent executes another pattern $p'_s$ at the end of which we ensure that the later agent has completed the execution of the $s$th basic pattern of its own sequence. Whether $p_s$ is $\fl$ or $\bo$, $p'_s$ is the same so that the earlier agent does not need to know the type of the $s$th basic pattern in the sequence of the later agent in order to push it (and by extension, does not require the knowledge of the label of the later agent). More precisely, $p'_s$ is chosen as follows.}

%{Thus, Procedure~$\pp$ may seem to require that each agent knows the exact sequence of calls to basic patterns in which Procedure~$\ite$ can be decomposed. Since the second part of this procedure consists in processing the transformed label, knowing the exact sequence of calls would require that each agent knows the label of the other agent, which is not the case. However, remark that the second part of every phase can be viewed as an alternation of Pattern~$\rd$ and either Pattern~$\fl$ or Pattern~$\bo$. The transformed label only determines which of $\fl$ and $\bo$ is executed in each pair of this alternation. Not knowing which of these patterns the later agent precisely executes is not an actual issue. When pattern $p_1$ in the sequence can be either $\fl$ or $\bo$, we choose the same pattern $p_2$ to ensure the completion of $p_1$ by the later agent. Thus, the agents do not need to know the label of the other one and thus the exact sequence of calls in which the execution of $\ite$ by the other agent can be decomposed. Precisely, in order to ensure the completion of $p_1$ by the later agent, we choose $p_2$ as follows.} 

{If $p_s$ is either Pattern~$\fl$ or Pattern~$\bo$, then $p'_s$ is Pattern~$\rd$: we use the same idea here as for the first synchronization mechanism. If $p_s$ is Pattern~$\rd$, then $p'_s$ is Pattern~$\fl$}, {relying on a property of the route $X\overline{X}$ (with $X$ any non-empty route) introduced in the last paragraph of Subsection~\ref{hood}: if both agents follow this route {concurrently} from the same node, then they meet. Pattern~$\di$ can be seen as such a route, and Procedure~$\fl$ (whose code is shown in Algorithm~\ref{bas:fl}) consists in executing Pattern~$\di$ from each node at distance at most $\alpha$. Hence, unless they meet, the later agent completes its execution of Pattern~$\rd$ before the earlier one starts executing $\di$ from the same node.} Note that $\pp$ uses as many patterns as the number of basic patterns in the sequence it is supposed to push: this and the fact of doubling the value of the input parameter of Procedure~$\ite$ in Algorithm~\ref{AsynchronousUnknownDistance} contribute in particular to keep the polynomiality of our solution.}

%$\ite$ whose execution it forces. Besides, doubling the assumption made on the upperbound on $D$ and $j$ allows to find a good one in a logarithmic in $D$ and $j$ number of attempts. Hence, the first part of $\p$ only implies doubling the number of executed patterns a logarithmic number of times, which makes our algorithm stay polynomial.}

{Thus, once the earlier agent completes the first part of $\p(\alpha)$, the later one has at least started the execution of $\ite(\alpha)$ (and thus of the first part of $\p(\alpha)$). At this point, we might think at first glance that we just shifted the problem. Indeed, the number of edge traversals that has to be made to complete all the executions of $\ite$ prior to $\ite(\alpha)$ is quite the same, if not higher, than the number of edge traversals that has to be made when executing the first part of $\p(\alpha)$. Hence the difference between both agents in terms of edge traversals has not been improved here. However, a crucial and decisive progress has nonetheless been done: contrary \emph{a priori} to the series of $\ite$ executed before {$\ite(\alpha)$}, the first part of $\p(\alpha)$ can be pushed at low cost via the execution of Pattern $\bo$ (line~4 of Algorithm~\ref{Pprocedure}) by the earlier agent. Actually this pattern corresponds to the kind of route, described at the end of Subsection~\ref{hood} for the second synchronization mechanism, which is of small length compared to the sequence of patterns it can push. Indeed, the first part of $\p(\alpha)$ can be viewed as a ``large'' sequence of Patterns~$\di$ and~$\fl$: however $\di$ and $\fl$ can be seen (by analogy with Subsection~\ref{hood}) as routes of the form $A\overline{A}$ and $B\overline{B}$ respectively, while Pattern $\bo$ executes $\di$ and $\fl$ (i.e., $A\overline{A}B\overline{B}$) once from at least each node at distance at most $\alpha$.} 

{Note that when the earlier agent has completed the execution of $\bo$ in $\p(\alpha)$, the later agent has at least started the execution of Pattern $\bo$ in $\p(\alpha)$. Hence, there is still a difference between both agents, but it has been considerably reduced: it is now relatively small so that we can handle it pretty easily afterwards.}

\begin{algorithm}[H]
		\caption{$\p(d)$} \label{Pprocedure}
		\begin{footnotesize}
		\begin{algorithmic}[1]
			\FOR {$i \leftarrow 1$; $i < d$; $i \leftarrow 2i$}
				\STATE Execute $\pp(i, d)$
			\ENDFOR
			\STATE {Execute $\bo(2d^4, d, d, 0)$}
			\STATE {Execute $\rd(2d^4+3d, C(\bo(2d^4, d, d, 0)))$} \label{lastSnailSynchronisation}
		\end{algorithmic}
		\end{footnotesize}
	\end{algorithm}

%{However, before processing the bits of its transformed label, the later agent could still have $\p(d)$ left to execute. This is where the second part of $\p(d)$ comes into the picture: it forces the later agent to execute the first part of $\p(d)$ (corresponding to the largest part of $\p(d)$) without escalation in the number of edge traversals. %Note that using here the idea of the first synchronization mechanism would lead to a superpolynomial cost as explained in Section~\ref{sec:idea}, and using $\pp$ once more would be of no help as this would simply displace the problem (indeed, the number of patterns to execute it would generate would be the same as in the previous assumptions. 
%Line
%However, thanks to the route described in Subsection~\ref{hood} for the second synchronization mechanism, these two escalations can be avoided with a single pattern. The first part of $\p(d)$ can be viewed as a sequence of Patterns~$\di$ and~$\fl$. Besides, $\di$ and $\fl$ can be seen as the routes $A\overline{A}$ and $B\overline{B}$. Hence, a single pattern that we call $\bo$ and which consists in executing $\di$ and $\fl$ once from each node at distance at most $d$ is enough to force the later agent to execute the whole first part of $\p(d)$.}

%\vspace{+0.5cm}
\begin{definition}[Basic and Perfect Decomposition] \label{def:decompo}
Given a call $P$ to an algorithm, we say that the basic decomposition of $P$, denoted by $\mathcal{BD}(P)$, is $P$ itself if $P$ corresponds to a basic pattern, the type of which belongs to $\{\rd;\fl;\bo\}$. Otherwise, {if $P$ contains no call or contains a moving instruction outside of every call} then $\mathcal{BD}(P)=\perp$, else $\mathcal{BD}(P)=\mathcal{BD}(x_1),\mathcal{BD}(x_2),\ldots,\mathcal{BD}(x_n)$ where $x_1,x_2,\ldots,x_n$ is the sequence (in the order of execution) of all the calls in $P$ that are children of $P$. We say that $\mathcal{BD}(P)$ is a perfect decomposition if it does not contain any $\perp$.
\end{definition}
%\vspace{-0.2cm}
\begin{remark}
\label{rem:per}
The basic decomposition of every call to Procedure~$\ite$ is perfect.
\end{remark}
%\vspace{-0.2cm}

\begin{algorithm}[H]
		\caption{$\pp(i, d)$} \label{patternPP}
		\begin{footnotesize}
		\begin{algorithmic}[1]
			\FOR {each $p$ in $\mathcal{BD}(\ite(i))$}
				\IF{$p$ is a call to pattern $\rd$ with value $\dip$ as first parameter}
					\STATE Execute $\fl(\dip, d)$
				\ELSE
					\STATE /* pattern $p$ is either a call to pattern $\fl$ or a call to pattern $\bo$ (in view of the above remark) and has at least two parameters */
					\STATE Let $\dip$ (resp. $\flp$) be the first (resp. the second) parameter of $p$
					\STATE Execute $\rd(d + \dip + 2\flp, C(\bo(\dip, \flp, \flp, 0)))$
				\ENDIF
			\ENDFOR 
		\end{algorithmic}
		\end{footnotesize}
	\end{algorithm}

\section{Proof of correctness and cost analysis}
\label{sec:proof}

	The purpose of this section is to prove that Algorithm~\aud{} ensures rendezvous in the basic grid at cost $\in\mathrm{O}((\dist+\sholab)^{33})$ with $\dist$ the initial distance between the agents and $\sholab$, the length of the shortest
	label. To this end, the section is made of four subsections. The first two subsections are dedicated to
	technical results about the basic patterns presented in Section~\ref{sec:bas} and synchronization properties
	of Algorithm \aud{}, which are used in turn to carry out the proof of correctness and the cost analysis of
	Algorithm \aud{} that are presented in the last two subsections.

\subsection{Properties of the basic patterns} \label{sec:lem}

	This subsection is dedicated to the presentation of some technical {results} about the basic patterns described in Section~\ref{sec:bas}. They are used in the following subsections to prove the correctness of Algorithm~\ref{AsynchronousUnknownDistance}.

	\subsubsection{Vocabulary}

		Before going any further, we need to introduce some extra vocabulary in order to facilitate the presentation of the next properties and lemmas.

		\begin{definition}
		\label{def:prec}
			A pattern execution $A$ precedes another pattern execution $B$ iff the beginning of $A$ occurs {by} the beginning of $B$.
		\end{definition}

		\begin{definition}
			Two pattern executions $A$ and $B$ are concurrent iff:
			\begin{itemize}
				\item pattern execution $A$ does not finish before pattern execution $B$ starts
				\item pattern execution $B$ does not finish before pattern execution $A$ starts
			\end{itemize}
		\end{definition}

		By misuse of {language}, in the rest of this paper we will sometimes say ``a pattern'' instead of ``a pattern execution''.

		Hereafter we say that a pattern $A$ \emph{concurrently precedes} a pattern $B$, iff $A$ and $B$ are concurrent, and $A$ precedes $B$.

		\begin{definition} \label{def:push}
			A pattern $A$ pushes a pattern $B$ if for every execution in which {$B$ precedes $A$}, agents meet before the end of the execution of {$A$} or $B$ finishes before $A$.
		\end{definition}

		In the sequel, given two sequences of moving instructions $X$ and $Y$, we will say that $X$ is a prefix of $Y$ if $Y$ can be viewed as the execution of the sequence $X$ followed by another {(possibly empty) sequence}.

	\subsubsection{Pattern $\di$}\label{sss:di}

		In this {section}, we show some properties related to Pattern~$\di$.

		Proposition~\ref{lem:di:lem} follows by induction on the input parameter of Pattern~$\di$ and Proposition~\ref{lem:di:pre} follows from Algorithm~\ref{bas:di}.

		\begin{proposition} \label{lem:di:lem}
			Let $\dip$ be any positive integer. Starting from a node $v$, Pattern~$\di(\dip)$ guarantees the following properties:
			\begin{enumerate}
				\item it allows to visit all nodes of the grid at distance at most $\dip$ from $v$
				\item it allows to traverse all edges of the grid linking two nodes at distance at most $\dip$ from $v$
%				\item it never makes the agent go at a distance greater than $\dip$ from $v$
%				\item its second period begins in the most northern node at distance $\dip$ from $v$
			\end{enumerate}
		\end{proposition}

		\begin{proposition} \label{lem:di:pre}
			Given two integers $\dip_1 \leq \dip_2$, the first period of Pattern $\di(\dip_1)$ is a prefix of the first period of Pattern $\di(\dip_2)$.
		\end{proposition}

		\begin{lemma} \label{lem:di:rdv}
			Let $\dip_1$ and $\dip_2$ be two positive integers such that $\dip_1 \leq \dip_2$. Let $a_1$ and $a_2$ be two agents executing respectively Patterns $\di(\dip_1)$ and $\di(\dip_2)$ both from the same node such that the execution of Pattern $\di(\dip_1)$ concurrently precedes the execution of Pattern $\di(\dip_2)$. Let $t_1$ (resp. $t_2$) be the time when agent $a_1$ (resp. $a_2$) completes the execution of Pattern $\di(\dip_1)$ (resp. $\di(\dip_2)$). Agents $a_1$ and $a_2$ meet by time $min(t_1, t_2)$.
		\end{lemma}

		\begin{proof}
%			Consider a node $v$ and a first agent $a_1$ executing Pattern $\di(\dip_1)$ from $v$ with $\dip_1$ any positive integer. Suppose that the execution of $\di(\dip_1)$ by $a_1$ concurrently precedes the execution of Pattern $\di(\dip_2)$ by another agent $a_2$ still from node $v$ with $\dip_1 \leq \dip_2$.

			In view of Proposition~\ref{lem:di:pre}, the first period of $\di(\dip_1)$ is a prefix of the first period of Pattern $\di(\dip_2)$. If the path followed by agent $a_1$ during its execution of $\di(\dip_1)$ is $e_1, e_2,\ldots, e_n, \linebreak \overline{e_1, e_2,\ldots, e_n}$ (the overlined part of the path corresponds to the backtrack), then the path followed by agent $a_2$ during the execution of Pattern $\di(\dip_2)$ is $e_1, e_2,\ldots, e_n, s, \overline{e_1, e_2,\ldots, e_n, s}$ where $s$ corresponds to the edges traversed at a distance $\in \{\dip_1 + 1; \ldots; \dip_2\}$.

{We have two cases to consider. If agent $a_2$ completes $e_1, e_2,\ldots, e_n$ by the time $a_1$ completes $e_1, e_2,\ldots, e_n$, then agents $a_1$ and $a_2$ meet while they are following $e_1, e_2,\ldots, e_n$ as agent $a_1$ is the first agent that starts following $e_1, e_2,\ldots, e_n$. Otherwise, agent $a_1$ starts following $\overline{e_1, e_2,\ldots, e_n}$ while $a_2$ is still following $e_1, e_2,\ldots, e_n$: this implies that the agents meet by the time $a_1$ (resp. $a_2$) finishes $\overline{e_1, e_2,\ldots, e_n}$ (resp. $e_1, e_2,\ldots, e_n$). So, in both cases the agents meet by time $min(t_1, t_2)$, which concludes the proof of this lemma.}

%{Let us distinguish two cases depending on which agent is the first to complete $e_1, e_2,\ldots, e_n$. If $a_2$ is this agent, since $a_1$ is the first agent which starts these edge traversals, they must have met while following them. In the other case, $a_1$ starts following $\overline{e_1, e_2,\ldots, e_n}$ while $a_2$ is still following $e_1, e_2,\ldots, e_n$ which means that they meet in this case too.}

%			Thus, if the execution of $\di(\dip_1)$ by $a_1$ concurrently precedes the execution of $\di(\dip_2)$ by agent $a_2$ both executed from the same node, agents meet by the end of these executions.
		\end{proof}

	\subsubsection{Pattern $\rd$}

		This {section} is dedicated to some properties of Pattern $\rd$. Informally speaking, Lemmas~\ref{lem:rd:fl} and~\ref{lem:rd:bo} describe the fact that Pattern $\rd$ pushes respectively Pattern $\fl$ and $\bo$ when it is given appropriate parameters.

		\begin{lemma} \label{lem:rd:fl}
			Consider two nodes $v_1$ and $v_2$ separated by a distance $\delta$. Let {$\fl(\dip_1, \flp)$ and $\rd(\dip_2, n)$} be two patterns respectively executed from $v_1$ and $v_2$ with $\dip_1$, $\dip_2$, $\flp$ and $n$ positive integers. If $\dip_2 \geq \dip_1 + \flp + \delta$ and $n \geq C(\fl(\dip_1, \flp))$ then Pattern $\rd(\dip_2, {n})$ pushes Pattern $\fl(\dip_1, \flp)$.
		\end{lemma}

		\begin{proof}
%			Assume that, in the grid, there are two agents $a_1$ and $a_2$. Denote by $v_1$ and $v_2$ their respective initial positions. Suppose that $v_1$ and $v_2$ are separated by a distance $\delta$. Assume that agent $a_1$ starts executing {Pattern $\fl(\dip_1, \flp)$} from node $v_1$ and agent $a_2$ performs {Pattern $\rd(\dip_2, n)$} on node $v_2$ (with $n \geq C(\fl(\dip_1, \flp))$ and $\dip_2 \geq \dip_1 + \flp + \delta$).

			{Denote by $a_1$ and $a_2$ the agents executing respectively $\fl(\dip_1, \flp)$ and $\rd(\dip_2, n)$.} Let us suppose by contradiction that $\rd(\dip_2, n)$ does not push $\fl(\dip_1, \flp)$, which means, by Definition~\ref{def:push} that {there exists an execution in which Pattern $\fl(\dip_1, \flp)$ precedes Pattern $\rd(\dip_2, n)$ such that} {$a_1$ neither meets $a_2$ nor completes $\fl(\dip_1, \flp)$ before $a_2$ completes $\rd(\dip_2, n)$}. {Remark that this implies in particular that these patterns are concurrent.}

			When executing its $\fl(\dip_1, \flp)$ agent {$a_1$} cannot be at a distance greater than $\dip_1 + \flp$ from its initial position {$v_1$} and thus cannot be at a distance greater than $\delta + \dip_1 + \flp$ from node {$v_2$}. {Also}, in view of Proposition~\ref{lem:di:lem}, each Pattern $\di(\dip_2)$ executed from node {$v_2$} which composes Pattern $\rd(\dip_2, n)$ allows to visit all nodes and to traverse all edges at distance at most $\dip_2$ from node {$v_2$}. Thus, each Pattern $\di(\dip_2)$ executed from node {$v_2$} allows to visit all nodes and to traverse all edges (although not necessarily in the same order) that are traversed during the execution of Pattern $\fl(\dip_1, \flp)$ from node {$v_1$}.

			Consider {the number of edge traversals completed by agent $a_1$ between the moment} when {$a_2$} starts executing any of the $\di(\dip_2)$ which compose $\rd(\dip_2, n)$ and {the moment} when {$a_2$} completes {this $\di(\dip_2)$}. If {$a_1$} has not completed a single edge traversal, then whether it was in a node or traversing an edge, it has met {$a_2$} which traverses every edge {$a_1$} traverses during its execution of $\fl(\dip_1, \flp)$. This contradicts our hypothesis, which implies that each time {$a_2$} completes one of its executions of Pattern $\di(\dip_2)$, {$a_1$} has completed at least one edge traversal. Since agent {$a_2$} executes $n \geq C(\fl(\dip_1, \flp))$ times Pattern $\di(\dip_2)$, {$a_1$} traverses at least $C(\fl(\dip_1, \flp))$ edges before {$a_2$} finishes executing its $\rd(\dip_2, n)$. As $C(\fl(\dip_1, \flp))$ is the number of edge traversals in $\fl(\dip_1, \flp)$, when {$a_2$} finishes executing Pattern $\rd(\dip_2, n)$, {$a_1$} has finished executing its Pattern $\fl(\dip_1, \flp)$, which contradicts our assumption and proves the lemma.
		\end{proof}
		
		{Using similar arguments to those used in the proof of Lemma~\ref{lem:rd:fl}, we can prove the following lemma.}

		\begin{lemma} \label{lem:rd:bo}
			Consider two nodes $v_1$ and $v_2$ separated by a distance $\delta$. Let {$\bo(\dip_1, \flp, \bop, \bops)$ and $\rd(\dip_2, n)$} be two patterns respectively executed from $v_1$ and $v_2$ with $\dip_1$, $\dip_2$, $\flp$, $\bop$, $\bops$ and $n$ positive integers. If $\dip_2 \geq \dip_1 + \flp + \bop + \delta$ and $n \geq C(\bo(\dip_1, \flp, \bop, \bops))$ then Pattern $\rd(\dip_2, n)$ pushes Pattern $\bo(\dip_1, \flp, \bop, \bops)$.
		\end{lemma}

	\subsubsection{Pattern $\fl$} \label{sss:fl}

		This {section} is dedicated to the properties of Pattern $\fl$. Informally speaking, Lemma~\ref{lem:fl:di} describes the fact that Pattern $\fl$ permits to push Pattern $\rd$ when it is given appropriate parameters. Proposition~\ref{lem:fl:pre} and Lemma~\ref{lem:fl:rdv} are respectively analogous to Proposition~\ref{lem:di:pre} and Lemma~\ref{lem:di:rdv}.

		In view of Algorithm~\ref{bas:fl}, we have the following proposition.

		\begin{proposition} \label{lem:fl:pre}
			Given four positive integers $\dip_1 + \flp_1 \leq \dip_2 + \flp_2$, the first period of $\fl(\dip_1, \flp_1)$ is a prefix of the first period of $\fl(\dip_2, \flp_2)$.
		\end{proposition}

		\begin{lemma} \label{lem:fl:di}
			Consider two nodes $v_1$ and $v_2$ separated by a distance $\delta$. Let {$\rd(\dip_1, n)$ and $\fl(\dip_2, \flp)$} be two patterns respectively executed from $v_1$ and $v_2$ with $\dip_1$, $\dip_2$, $\flp$ and $n$ positive integers. If $\flp \geq \delta$ and $\dip_1 \leq \dip_2$ then Pattern $\fl(\dip_2, \flp)$ pushes Pattern $\rd(\dip_1, n)$.
		\end{lemma}

		\begin{proof}
%			Assume that there are two agents $a_1$ and $a_2$ initially separated by a distance $\delta$. Assume that their respective initial positions are node $v_1$ and node $v_2$. Agent {$a_1$} executes Pattern $\rd(\dip_1, n)$ from {$v_1$} with $\dip_1$ and $n$ any positive integers, {and agent $a_2$ executes Pattern $\fl(\dip_2, \flp)$ from $v_2$ with $\flp \geq \delta$ and $\dip_2 \geq \dip_1$.}

			{Denote by $a_1$ and $a_2$ the agents executing respectively $\rd(\dip_1, n)$ and $\fl(\dip_2, \flp)$.} Let us suppose by contradiction that $\fl(\dip_2, \flp)$ does not push $\rd(\dip_1, n)$ which means by Definition~\ref{def:push} that {there exists an execution in which Pattern $\rd(\dip_1, n)$ precedes Pattern $\fl(\dip_2, \flp)$ such that} {$a_1$ neither meets $a_2$ nor completes $\rd(\dip_1, n)$ before $a_2$ completes $\fl(\dip_2, \flp)$}. When executing $\fl(\dip_2, \flp)$, agent {$a_2$} performs $\di(\dip_2)$ from each node at distance at most $\flp$ from {$v_2$} with $\flp \geq \delta$. Thus, at some point, {$a_2$} executes $\di(\dip_2)$ from node {$v_1$}. In view of Lemma~\ref{lem:di:rdv}, since $\dip_2 \geq \dip_1$ and since by assumption, {$a_1$} has not finished executing its $\rd(\dip_1, n)$ when {$a_2$} starts executing Pattern $\di(\dip_2)$ from {$v_1$}, agents meet by the end of the latter {and thus before the end of $\fl(\dip_2, \flp)$} which contradicts our assumption and proves the lemma.
		\end{proof}

		\begin{lemma} \label{lem:fl:rdv}
			Consider two agents $a_1$ and $a_2$ executing respectively Patterns $\fl(\dip_1, \flp_1)$ and $\fl(\dip_2, \flp_2)$ both from node $v$ with $\dip_1$, $\dip_2$, $\flp_1$ and $\flp_2$ positive integers such that $\dip_2 + \flp_2 \geq \dip_1 + \flp_1$. Suppose that the execution of $\fl(\dip_1, \flp_1)$ by $a_1$ concurrently precedes the execution of $\fl(\dip_2, \flp_2)$ by $a_2$. Let $t_1$ (resp. $t_2$) be the time when agent $a_1$ (resp. $a_2$) completes its execution of Pattern $\fl(\dip_1, \flp_1)$ (resp. $\fl(\dip_2, \flp_2)$). Agents $a_1$ and $a_2$ meet by time $min(t_1, t_2)$.
		\end{lemma}

		\begin{proof}
%			Consider a node $v$ and a first agent $a_1$ executing Pattern $\fl(\dip_1, \flp_1)$ from $v$ with $\dip_1$ and $\flp_1$ two positive integers. Suppose that the execution of Pattern $\fl(\dip_1, \flp_1)$ by $a_1$ concurrently precedes an execution of Pattern $\fl(\dip_2, \flp_2)$ by another agent $a_2$ still from node $v$ with $\dip_2 + \flp_2 \geq \dip_1 + \flp_1$.

			This proof is similar to the proof of Lemma~\ref{lem:di:rdv}. In view of Proposition~\ref{lem:fl:pre}, if the path followed by agent $a_1$ during its execution of $\fl(\dip_1, \flp_1)$ is $e_1, e_2,\ldots, e_n, \overline{e_1, e_2,\ldots, e_n}$ (the overlined part of the path corresponds to the backtrack), then the path followed by agent $a_2$ during the execution of Pattern $\fl(\dip_2, \flp_2)$ is $e_1, e_2,\ldots, e_n, s, \overline{e_1, e_2,\ldots, e_n, s}$ where $s$ corresponds to the edges traversed from the $(\dip_1 + \flp_1 + 1)$-th iteration of the main loop of Pattern $\fl$ to its $(\dip_2 + \flp_2)$-th iteration. 
{We have two cases to consider. If agent $a_2$ completes $e_1, e_2,\ldots, e_n$ by the time $a_1$ completes $e_1, e_2,\ldots, e_n$, then agents $a_1$ and $a_2$ meet while they are following $e_1, e_2,\ldots, e_n$ as agent $a_1$ is the first agent that starts following $e_1, e_2,\ldots, e_n$. Otherwise, agent $a_1$ starts following $\overline{e_1, e_2,\ldots, e_n}$ while $a_2$ is still following $e_1, e_2,\ldots, e_n$: this implies that the agents meet by the time $a_1$ (resp. $a_2$) finishes $\overline{e_1, e_2,\ldots, e_n}$ (resp. $e_1, e_2,\ldots, e_n$). So, in both cases the agents meet by time $min(t_1, t_2)$, which concludes the proof of this lemma.}

%{Let us distinguish two cases depending on which agent is the first to complete $e_1, e_2,\ldots, e_n$. If $a_2$ is this agent, since $a_1$ is the first agent which starts these edge traversals, they must have met while following them. In the other case, $a_1$ starts following $\overline{e_1, e_2,\ldots, e_n}$ while $a_2$ is still following $e_1, e_2,\ldots, e_n$ which means that they meet in this case too.}

%			Let $t_1$ (resp. $t_2$) be the time when agent $a_1$ (resp. $a_2$) completes its execution of Pattern $\fl(\dip_1, \flp_1)$ (resp. $\fl(\dip_2, \flp_2)$). In the same way as in the proof of Lemma~\ref{lem:di:rdv}, if the execution of $\fl(\dip_1, \flp_1)$ by $a_1$ concurrently precedes the execution of $\fl(\dip_2, \flp_2)$ by agent $a_1$ both executed from the same node, the agents meet by time $min(t_1, t_2)$.
		\end{proof}

	\subsubsection{Pattern $\bo$}\label{sss:bo}

Informally speaking, the following lemma highlights the fact that Pattern $\bo$ can push ``a lot of basic patterns''
under some conditions. In other words, we can force an agent to make a lot of edge traversals ``at relative low cost''.

		\begin{lemma} \label{lem:bo}
			Consider two nodes $v_1$ and $v_2$ separated by a distance $\delta$. Let {$S$ and $\bo(\dip_1, \flp_1, \bop, \bops)$ be respectively a sequence of Patterns $\rd$ and $\fl$ executed from $v_1$ and a pattern executed from $v_2$ with $\dip_1$, $\flp_1$}, $\bop$ and $\bops$ four positive integers. If $\bop \geq \delta$ and for each Pattern $\rd$ $R$ and Pattern $\fl$ $B$ belonging to $S$, {$\dip_1 + \flp_1$} is greater than or equal to the sum of the parameters of $B$, and {$\dip_1$} is greater than or equal to the first parameter of $R$, then Pattern {$\bo(\dip_1, \flp_1, \bop, \bops)$} pushes $S$.
		\end{lemma}

		\begin{proof}
%			Let {$a_1$} be an agent executing a sequence $S$ of Patterns $\rd$ and $\fl$ from a node {$v_1$}. Suppose that there exist two integers $\dip_1$ and $\flp_1$ such that each Pattern $\fl$ $B$ inside the sequence is assigned parameters whose sum is at most $\dip_1 + \flp_1$, and such that each Pattern $\rd$ $R$ of the sequence is assigned a first parameter which is at most $\dip_1$. Let {$v_2$} be another node separated from {$v_1$} by a distance $\delta$. Suppose that another agent {$a_2$} executes Pattern $\bo(\dip_1, \flp_1, \bop, \bops)$ from {$v_2$} with $\bop \geq \delta$ and $\bops$ two positive integers.

			{Denote by $a_1$ and $a_2$ the agents executing respectively $S$ and $\bo(\dip_1, \flp_1, \bop, \bops)$.} In order to prove that the execution of Pattern $\bo(\dip_1, \flp_1, \bop, \bops)$ by {$a_2$} pushes the sequence of {patterns} $S$, let us suppose by contradiction that {there exists an execution in which $S$ precedes Pattern $\bo(\dip_1, \flp_1, \bop, \bops)$ such that} {$a_1$ neither meets $a_2$ nor completes its whole sequence of patterns before $a_2$ completes $\bo(\dip_1, \flp_1, \bop, \bops)$}.

			In view of Algorithm~\ref{bas:bo}, when executing $\bo(\dip_1, \flp_1, \bop, \bops)$, {$a_2$} executes Pattern $\di(\dip_1)$ followed by Pattern $\fl(\dip_1, \flp_1)$ on each node at distance at most $\bop$ from {$v_2$}. Since $\bop \geq \delta$, during its execution of $\bo(\dip_1, \flp_1, \bop, \bops)$, {$a_2$ follows $\pa(v_2, v_1)$}, executes Pattern $\di(\dip_1)$ (denoted by $p_1$) and then Pattern $\fl(\dip_1, \flp_1)$ (denoted by $p_2$) both from node {$v_1$}. In order to prove that the execution of $\bo(\dip_1, \flp_1, \bop, \bops)$ by {$a_2$} pushes the execution of $S$ by {$a_1$}, we are going to prove that the agents meet by the time {$a_2$} completes its executions of $p_1$ and $p_2$.

			By assumption, {$a_1$} has not finished executing $S$ when {$a_2$} arrives on {$v_1$} to execute $p_1$ and $p_2$. Let us consider what it can be executing at this moment. If it is executing Pattern $\di(\dip_2)$ with $\dip_2 \leq \dip_1$ a positive integer, then in view of Lemma~\ref{lem:di:rdv}, the agents meet by the end of the execution of $p_1$, which contradicts the assumption that the agents do not meet {before} the end of $\bo(\dip_1, \flp_1, \bop, \bops)$. This means that when {$a_2$} starts executing $p_1$, {$a_1$} is executing Pattern $\fl(\dip_2, \flp_2)$ for some positive integers $\dip_2$ and $\flp_2$ such that $\dip_2 + \flp_2 \leq \dip_1 + \flp_1$. After $p_1$, {$a_2$} executes {$p_2$}. By Lemma~\ref{lem:fl:rdv}, if {$a_1$} is still executing Pattern $\fl(\dip_2, \flp_2)$ for some positive integers $\dip_2$ and $\flp_2$ such that $\dip_2 + \flp_2 \leq \dip_1 + \flp_1$ (the same as above, or another) then the agents meet by the end of the execution of $p_2$ which contradicts our assumption once again. As a consequence, when {$a_2$} starts executing $p_2$, {$a_1$} is executing Pattern $\di(\dip_3)$ for some positive integer $\dip_3 \leq \dip_1$. Denote by $p_3$ this pattern, and remember that {$a_1$} {starts it after} {$a_2$} starts $p_1$. Moreover, when {$a_2$} starts executing $p_2$, {$a_1$} can not be in {$v_1$} as it is the node where {$a_2$} starts $p_2$, thus it has at least {started the first edge traversal} of $p_3$. Hence, $p_1$ concurrently precedes $p_3$, and {$a_2$ completes the execution of $p_1$ before $a_1$ completes the execution of $p_3$}.

			In view of Algorithm~\ref{bas:di}, like in the proof of Lemma~\ref{lem:di:rdv}, we can denote by $e_1, \ldots, e_n, \overline{e_1, \ldots, e_n}$ the route followed by {$a_1$} when executing $p_3$ and by $e_1, \ldots, e_n, s, \overline{e_1, \ldots, e_n, s}$ the route followed by {$a_2$} when executing $p_1$ where $s$ corresponds to edges traversed at a distance belonging to {$\{\dip_3 + 1; \ldots; \dip_1\}$}. Remark that in view of the definition of a backtrack, $\overline{e_1, \ldots, e_n, s} = \overline{s}, \overline{e_1, \ldots, e_n}$. Consider the moment $t_1$ when {$a_1$} completes the first period of $p_3$ and begins the second one. It has just traversed $e_1, \ldots, e_n$, and is about to follow $\overline{e_1, \ldots, e_n}$. At this moment, {$a_2$} can not have {started the edge traversals $\overline{e_1, \ldots, e_n}$}, or else agents have met by $t_1$, which would contradict our assumption. However, as $p_1$ is completed before $p_3$, {$a_2$} must finish executing {some non-empty part of $s\overline{s}$ followed by $\overline{e_1, \ldots, e_n}$} before {$a_1$} finishes executing $\overline{e_1, \ldots, e_n}$ which implies that {the} agents meet by the end of the execution of $p_1$ and contradicts once again the hypothesis that they do not meet by the end of $p_2$.

			So, in every case, the assumption that {before the end of the execution of $\bo(\dip_1, \flp_1, \bop, \bops)$, {$a_1$} neither meets {$a_2$} nor finishes executing $S$} is contradicted. Hence, the execution of Pattern $\bo(\dip_1, \flp_1, \bop, \bops)$ by {$a_2$} pushes the execution of $S$ by {$a_1$}, and the lemma holds.
		\end{proof}

	\subsection{Agents synchronizations}
\label{sub:sync}

		We recall the reader that $\dist$ is the initial distance separating the two agents in the basic grid.

		The aim of this subsection is to introduce and prove several synchronization properties our algorithms offer (cf., Lemmas~\ref{lemmaFireworks} and~\ref{lemmaBitSynchronization}). By ``synchronization" we mean that if one agent has completed some part of its rendezvous algorithm, then either it must have met the other agent or this other agent has also completed some part (not necessarily the same one) of its algorithm \ie it must have made progress.

		To prove Lemmas~\ref{lemmaFireworks} and~\ref{lemmaBitSynchronization}, we first need to show some
		more technical results---Lemmas~\ref{cor:syn:pp1},~\ref{cor:syn:pp2}, and~\ref{cor:syn:ite}.

%		\begin{lemma} \label{cor:syn:pp1}
%			Let $v_1$ and $v_2$ be the two nodes initially occupied by the agents $a_1$ and $a_2$ respectively. {Let us assume agent $a_1$ and $a_2$ respectively execute the procedures sequences $S_1 = \ite(1), \dots, \ite(2^{c_1})$ and $S_2 = \pp(1, d_1), \dots, \pp(2^{c_1}, d_1)$ with $c_1$ and $d_1 \geq \dist$ some non-negative integers. Either the agents meet by the end of the execution of $S_2$ or $S_1$ finishes before $S_2$.}
%		\end{lemma}

		\begin{lemma} \label{cor:syn:pp1}
			Let $v_1$ and $v_2$ be the two nodes separated by a distance $\dist$ that are initially occupied by the agents $a_1$ and $a_2$ respectively. {Let $c_1$ and $d_1$ be two non-negative integers such that $d_1 \geq \dist$. Assume the prefix of the execution of agent $a_1$ is the sequence $S_1 = \ite(1), \dots, \ite(2^{c_1})$. Assume that a part of the execution of agent $a_2$ is the sequence $S_2 = \pp(1, d_1), \dots, \pp(2^{c_1}, d_1)$. Either the agents meet {before} the end of the execution of $S_2$ or $S_1$ finishes before $S_2$}.
		\end{lemma}
		
		\begin{proof}
			%Consider two agents $a_1$ and $a_2$. Their respective initial nodes are $v_1$ and $v_2$, which are separated by a distance $\dist$. {Suppose that $a_1$ and $a_2$ respectively execute $S_1 = \ite(1), \dots, \ite(2^{c_1})$ and $S_2 = \pp(1, d_1), \dots, \pp(2^{c_1}, d_1)$ with $c_1$ and $d_1 \geq \dist$ some non-negative integers. 

{Assume by contradiction that there exists some scenario $E_1$ in which neither the agents meet {before} the end of the execution of $S_2$ by agent $a_2$ nor the execution of $S_1$ by $a_1$ finishes before the execution of $S_2$ by $a_2$.}
			
			In view of Algorithm~\ref{patternPP}, {and since there are as many occurrences of Procedure $\ite$ in $S_1$ as of Procedure $\pp$ in $S_2$}, there are as many basic patterns (from $\{\rd; \fl; \bo\}$) in $\mathcal{BD}({S_1})$ as in $\mathcal{BD}({S_2})$. Each basic pattern inside $\mathcal{BD}({S_1})$ and $\mathcal{BD}({S_2})$ is given an index between 1 and $n$ according to its order of appearance. {In view of} Remark~\ref{rem:per}, {for any integer $d_2$, $\mathcal{BD}(\ite(d_2))$} is perfect, {which implies that $\mathcal{BD}(S_1)$ is perfect too.} This has the following consequences. When agent {$a_1$} starts the execution of {$S_1$}, this agent starts the execution of the first basic pattern in $\mathcal{BD}({S_1})$. {Moreover}, when agent {$a_1$} completes the execution of {$S_1$}, it completes the execution of the $n$-th basic pattern in $\mathcal{BD}({S_1})$. Lastly, for any integer $i$ between 1 and $n - 1$, agent {$a_1$} does not make any edge traversal between the $i$-th and the $(i + 1)$-th basic pattern in $\mathcal{BD}({S_1})$. In other words, every edge traversal agent {$a_1$} makes during the execution of {$S_1$} is performed during one of the basic patterns inside $\mathcal{BD}({S_1})$. Remark that $\mathcal{BD}({S_2})$ is {perfect too}.
			
			{Let us show by induction on $i$} that for {every} integer $i$ between 1 and $n$, {$a_1$ either meets $a_2$ or completes the execution of the $i$-th pattern inside $\mathcal{BD}({S_1})$ before $a_2$ completes the execution of the $i$-th pattern inside $\mathcal{BD}({S_2})$}. {If $i = 1$, we distinguish two cases.} {In the first case, the first pattern of $\mathcal{BD}(S_2)$ starts before the first pattern of $\mathcal{BD}(S_1)$, while in the second case it does not \ie, in view of Definition~\ref{def:prec}, the first pattern of $\mathcal{BD}(S_1)$ precedes the first pattern of $\mathcal{BD}(S_2)$.}
			
			{In the first case, since it does not make any edge traversal before the moment $t_1$ when it starts executing the first pattern of $\mathcal{BD}(S_1)$, we assume that $a_1$ is in $v_1$ from the moment $t_2$ when $a_2$ starts executing the first pattern in $\mathcal{BD}(S_2)$ to $t_1$. We can build another {scenario} $E_2$ in which $a_1$ (resp. $a_2$) executes $S_1$ (resp. $S_2$) from $v_1$ (resp. $v_2$) as in $E_1$, at every moment of $E_2$ both $a_1$ and $a_2$ are at the exact same place as in $E_1$, but in which the first pattern of $\mathcal{BD}(S_1)$ precedes the first pattern of $\mathcal{BD}(S_2)$. We achieve this by designing the behavior of the adversary in $E_2$ as follows. The adversary handles $a_2$ in the same way in $E_2$ as in $E_1$. From the moment $t_3$ at which $a_1$ starts executing $S_1$ in $E_2$ to the moment $t_2$ at which $a_2$ starts executing $S_2$ (both in $E_1$ and $E_2$), as well as from $t_2$ to the moment $t_1$ at which $a_1$ starts executing $S_1$ in $E_1$, in $E_2$, the adversary prevents $a_1$ from moving from $v_1$. Moreover, from $t_1$ on, in $E_2$, the adversary handles $a_1$ as in $E_1$. Since at every moment both $a_1$ and $a_2$ are at the same place in $E_1$ and $E_2$, if we prove in $E_2$ that {$a_1$ either meets $a_2$ or completes the execution of the first pattern inside $\mathcal{BD}({S_1})$ before $a_2$ completes the execution of the first pattern inside $\mathcal{BD}({S_2})$}, then this also holds in $E_1$. Also, in $E_2$, the first pattern of $\mathcal{BD}(S_1)$ precedes the first pattern of $\mathcal{BD}(S_2)$, this is the second of the two cases we distinguish.} {Hence, when $i=1$ it is enough to consider the second case only.}
			
			{If the first pattern of $\mathcal{BD}(S_1)$ precedes the first pattern of $\mathcal{BD}(S_2)$, then} in view of Lemmas~\ref{lem:rd:fl},~\ref{lem:rd:bo} and~\ref{lem:fl:di}, Algorithm~\ref{patternPP} and the fact that $d_1 \geq \dist$, whatever the type of the {first} pattern inside $\mathcal{BD}({S_1})$ ($\fl$, $\bo$ or $\rd$), {$a_1$ either meets $a_2$ or completes the first pattern inside $\mathcal{BD}(S_1)$ before $a_2$ completes the first pattern inside $\mathcal{BD}(S_2)$.}
			
			{Let us now assume that there exists an integer $j$ {in $\{1,\ldots,(n - 1)\}$} {such that $a_1$ either meets $a_2$ or completes the $j$-th pattern inside $\mathcal{BD}(S_1)$ before $a_2$ completes the $j$-th pattern inside $\mathcal{BD}(S_2)$} and show that {$a_1$ either meets $a_2$ or completes the $(j + 1)$-th pattern inside $\mathcal{BD}(S_1)$ before $a_2$ completes the $(j + 1)$-th pattern inside $\mathcal{BD}(S_2)$}. In order to achieve this, let us suppose that the agents do not meet {before} the end of the execution of the $(j + 1)$-th pattern inside $\mathcal{BD}(S_2)$ and show that the execution of the $(j + 1)$-th pattern inside $\mathcal{BD}(S_1)$ finishes before the execution of the $(j + 1)$-th pattern inside $\mathcal{BD}(S_2)$. In view of the induction hypothesis, and the assumption that the agents do not meet {before} the end of the execution of the $(j + 1)$-th pattern inside $\mathcal{BD}(S_2)$, the $j$-th pattern inside $\mathcal{BD}(S_1)$ finishes before the $j$-th pattern inside $\mathcal{BD}(S_2)$ which means that the $(j + 1)$-th pattern inside $\mathcal{BD}(S_1)$ precedes the $(j + 1)$-th pattern inside $\mathcal{BD}(S_2)$. Again, in view of Lemmas~\ref{lem:rd:fl},~\ref{lem:rd:bo} and~\ref{lem:fl:di}, Algorithm~\ref{patternPP} and the fact that $d_1 \geq \dist$, whatever the type of the $(j + 1)$-th pattern inside $\mathcal{BD}(S_1)$, it finishes before the $(j + 1)$-th pattern inside $\mathcal{BD}(S_2)$.} In particular, if the {$(j + 1)$}-th pattern inside $\mathcal{BD}({S_1})$ is a $\fl$ or a $\bo$ called after the test at line~\ref{readingTest}, at line~\ref{callFL} or~\ref{callBO}, {of Algorithm~\ref{iteration}}, {regardless of which of the two patterns it is, $a_1$ completes its execution {before} the end} of the {$(j + 1)$}-th pattern inside $\mathcal{BD}({S_2})$. Indeed, for any positive integers $\dip$, {$\flp$, $\bop$} and $\bops$, $\bo(\dip, {\flp, \bop}, \bops)$ can be viewed as composed of several $\fl(\dip, {\flp})$ so that $C(\bo(\dip, {\flp, \bop}, \bops)) \geq C(\fl(\dip, {\flp}))$.
			
			{This means in particular that {before} the end of the $n$-th pattern inside $\mathcal{BD}(S_2)$ and thus {before} the end of $S_2$, $a_1$ either meets $a_2$ or completes the $n$-th pattern inside $\mathcal{BD}(S_1)$ and thus $S_1$ itself, which completes the proof.}
		\end{proof}

		\begin{lemma} \label{cor:syn:pp2}
			Let {$d_1$ and $\dip_1$ be some integers} such that the first parameter of each basic pattern inside $\mathcal{BD}(\ite(d_1))$ is assigned a value which is at most $\dip_1$. For {every integer $d_2 \geq d_1$}, the first parameter of each basic pattern inside $\mathcal{BD}(\pp(d_1, d_2))$ is {less than} or equal to $\dip_1 + 3d_2$.
		\end{lemma}

		\begin{proof}
%			Make the assumption that there {exist two integers $d_1$ and $\dip_1$} such that the first parameter of each basic pattern inside $\mathcal{BD}(\ite(d_1))$ is given a value {less than} or equal to $\dip_1$.
			
			In view of Algorithm~\ref{patternPP}, each basic pattern inside $\mathcal{BD}(\ite(d_1))$ and $\mathcal{BD}(\pp(d_1, d_2))$ (with $d_2 \geq d_1$ {some integer}) is given an index {from 1 to $n$} according to its order of appearance, with $n$ the number of basic patterns in either of these decompositions. Thus, for any integer $i$ {from 1 to $n$}, there is a pair of patterns $(p_1, p_2)$ such that $p_1$ is the $i$-th basic pattern inside $\mathcal{BD}(\ite(d_1))$, and $p_2$ is the $i$-th pattern inside $\mathcal{BD}(\pp(d_1, d_2))$. Let us show that there is no such pair $(p_1, p_2)$ such that the first parameter of $p_2$ is given a value greater than $\dip_1 + 3d_2$. To this end, we analyse three cases depending on the type of pattern $p_1$.

			Let us first consider the case in which $p_1$ is Pattern $\rd(\dip_2, n_1)$ with $\dip_2 \leq \dip_1$ and $n_1$ two positive integers. In view of Algorithm~\ref{patternPP}, since $p_1$ is Pattern $\rd(\dip_2, n_1)$, $p_2$ is $\fl(\dip_2, d_2)$, which means that its first parameter is at most $\dip_1$ and thus at most $\dip_1 + 3d_2$.

			Let us now consider the cases in which $p_1$ is either Pattern $\fl$ or Pattern $\bo$. We first make the following remark. In $\mathcal{BD}(\ite(d_1))$, whether it is called directly by Procedure~$\ite(d_1)$, or inside its call to $\p(d_1)$, or inside the	call of the latter to $\pp(d_3, d_1)$ with some {integer} $d_3 < d_1$, the second parameter of Pattern $\fl$ is always $d_1$, and the second and third parameters of Pattern $\bo$ are always $d_1$ as well.
			
			{In view of Algorithm~\ref{patternPP}, whether $p_1$ is Pattern $\fl(\dip_2, d_1)$ or Pattern $\bo(\dip_2, d_1, d_1, \bops)$ with two positive integers $\bops$ and $\dip_2 \leq \dip_1$, $p_2$ is $\rd(d_2 + \dip_2 + 2d_1, C(\bo(\dip_2, d_1, d_1, \bops)))$.} Its first parameter is $d_2 + \dip_2 + 2d_1$ which is at most $\dip_1 + 3d_2$.
			
%			If $p_1$ is Pattern $\fl(\dip_2, d_1)$ with some integer $\dip_2 \leq \dip_1$, then in view of Algorithm~\ref{patternPP}, $p_2$ is $\rd(d_1 + d_2 + \dip_2, C(\fl(\dip_2, d_1)))$. Its first parameter is $d_1 + d_2 + \dip_2$ which is at most $\dip_1 + 2d_2 < \dip_1 + 3d_2$.
			
%			At last, in view of Algorithm~\ref{patternPP}, if $p_1$ is Pattern $\bo(\dip_2, d_1, d_1, \bops)$ with two positive integers $\bops$ and $\dip_2 \leq \dip_1$, $p_2$ is $\rd(d_2 + 2d_1 + \dip_2, C(\bo(\dip_2, d_1, d_1, \bops)))$. Its first parameter is $d_2 + 2d_1 + \dip_2$ which is at most $\dip_1 + 3d_2$.

			Hence, within $\mathcal{BD}(\pp(d_1, d_2))$, there cannot be any call to a basic pattern in which the first parameter is assigned a value greater than $\dip_1 + 3d_2$, which proves the lemma.
		\end{proof}

		\begin{lemma} \label{cor:syn:ite}
			The first parameter of each basic pattern inside $\mathcal{BD}(\ite(d_1))$ (with $d_1$ any power of two) is at most {$32d_1^4 - 6d_1$.}
		\end{lemma}

		\begin{proof}
			We prove this lemma by induction on $d_1$.

			{Let us first consider that $d_1 = 1$. We enumerate the basic patterns inside $\mathcal{BD}(\ite(1))$ and show that for each of them the first parameter is given a value which is {less} than or equal to $32d_1^4 - 6d_1 = 26$. Procedure $\ite(1)$ begins with $\p(1)$ which is composed of calls to $\bo(2, 1, 1, 0)$ and $\rd(5, C(\bo(2, 1, 1, 0)))$, with both first parameters lower than $26$. After $\p(1)$ too, the first parameter that is given to the patterns called in Procedure $\ite(1)$ is always at most $26$. Indeed, the first parameter is assigned its maximum value when $j = 2d_1 (d_1 + 1) = 4$ and $i = d_1 = 1$ \ie when $3d_1 = 3$ has been added $i(j + 1) = 5$ times to the initial value of $radius$ \ie $5$, which gives a maxiuml value equal to $5 + 15 = 20 < 26$.} This concludes the analysis of the case when $d_1 = 1$.

			 Let us now assume that there exists a power of two $d_2$ such that for each power of two $d_3\leq d_2$, the first parameter of each basic pattern inside $\mathcal{BD}(\ite(d_3))$ is at most {$32d_2^4 - 6d_2$. We once again enumerate the basic patterns inside $\mathcal{BD}(\ite(2d_2))$ and show that each of them is given a value for the first parameter which is at most $512d_2^4 - 12d_2$.} Procedure $\ite(2d_2)$ begins with $\p(2d_2)$ which in turn, begins with $\pp(1, 2d_2)$, \ldots, $\pp(d_2, 2d_2)$. By induction hypothesis, inside $\mathcal{BD}(\ite(1))$, \ldots, $\mathcal{BD}(\ite(d_2))$, the first parameter of each basic pattern is at most {$32d_2^4 - 6d_2$.} In view of Lemma~\ref{cor:syn:pp2}, inside $\mathcal{BD}(\pp(1, 2d_2))$, \ldots, $\mathcal{BD}(\pp(d_2, 2d_2))$, the first parameter of each basic pattern is at most {$32d_2^4 - 6d_2 + 6d_2 = 32d_2^4 < 512d_2^4 - 12d_2$.} Moreover, after $\pp(1, 2d_2)$, \ldots, $\pp(d_2, 2d_2)$, Procedure~$\p(2d_2)$ calls {Pattern $\bo(32d_2^4, 2d_2, 2d_2, 0)$} followed by {Pattern $\rd(32d_2^4+6d_2,C(\bo(32d_2^4, 2d_2, 2d_2, 0)))$.} Inside these calls, the first parameter is respectively given the values {$32d_2^4$ and $32d_2^4+6d_2$ which are both lower than $512d_2^4 - 12d_2$}. {Moreover}, after $\p(2d_2)$, in the same way as when $d_1 = 1$, we can show that the first parameter keeps increasing and reaches a maximum value equal to {$32d_2^4 + 6d_2 + 12d_2^2(4d_2 (2d_2 + 1) + 1) = 128d_2^4 + 48d_2^3 + 12d_2^2 + 6d_2 < 512d_2^4 - 12d_2$} which completes the proof of the lemma.
		\end{proof}

		Before presenting the next lemma, we need to introduce the following notions. We say that the first four lines of Algorithm $\p$ are its first part, and that the last line is the second part. Procedure $\ite$ begins with a call to Procedure $\p$: We will consider that the first part of Procedure $\ite$ is the first part of this call, and that the second part of Procedure $\ite$ is the second part of this call. After these two parts, there is a third part in Procedure $\ite$ which consists of calls to basic patterns. Moreover, note that the execution of Algorithm \aud{} can be viewed as a sequence of consecutive calls to Procedure $\ite$ with an increasing parameter. We will say that the $(i + 1)$-th call to Procedure $\ite$ {(\ie the call to Procedure $\ite(2 ^ i)$)} by an agent executing Algorithm \aud{} is Phase~$i$.

		\begin{lemma} \label{lemmaFireworks}
			Consider two agents {$a_1$ and $a_2$} executing Algorithm \aud{}. Let {$i_1$ and $d_1$ be two integers such that $2 ^ i_1 = d_1 \geq \dist$}. {Agent $a_1$ either meets $a_2$ or completes the execution of the first part of Phase~$i_1$ before agent $a_2$ completes the execution of the second part of Phase~$i$.}
		\end{lemma}

		\begin{proof}
%			Let $a_1$ and $a_2$ be two agents executing Algorithm \aud{}. Let $i_1$ and $d_1$ be two integers such that $2 ^ {i_1} = d_1 \geq \dist$
			
			Assume by contradiction that {the lemma is false}. {This implies in particular that when $a_2$ finishes executing the second part of Phase~$i_1$, $a_1$ is either executing Phase~$i_2$ for an integer $i_2 < i_1$, or the first part of Phase~$i_1$}.

			First of all, in view of Lemma~\ref{cor:syn:pp1} and since $d_1 \geq \dist$, we know that {$a_1$ either meets $a_2$ or fini\-shes executing the sequence $\ite(1)$, \ldots, $\ite(2 ^ {i_1 - 1})$ before $a_2$ completes the sequence $\pp(1, d_1)$, \ldots, $\pp(2 ^ {i_1 - 1}, d_1)$ (\ie the loop at the beginning of procedure $\p(d_1)$).} Given that by assumption, agents do not meet before {$a_2$} completes its execution of the second part of Phase~$i_1$, {$a_1$ starts executing the first part of Phase~$i_1$ before $a_2$ finishes executing the loop at the beginning of Procedure $\p(d_1)$}, {which means that the execution of the loop at the beginning of Procedure $\p(d_1)$ by $a_1$ precedes the execution of $\bo(2d_1^4, d_1, d_1, 0)$ by $a_2$.}

			Let us {build on this to} show that when {$a_2$} finishes executing $\bo(2d_1^4, d_1, d_1, 0)$, {$a_1$} has finished executing the loop at the beginning of Procedure $\p(d_1)$. In view of Lemmas~\ref{cor:syn:pp2} and~\ref{cor:syn:ite}, while executing this loop, {$a_1$} executes a sequence of Patterns $\rd$ and $\fl$ called by Procedure $\pp$ whose the first parameter is at most $2d_1^4$. Since $d_1 \geq \dist$, in view of Lemma~\ref{lem:bo} {and the assumption that the agents do not meet {before} the end of the execution of the second part of Phase~$i_1$ by $a_2$}, when {$a_2$} finishes executing $\bo(2d_1^4, d_1, d_1, 0)$, {$a_1$} has finished executing the loop.

			After executing {Pattern~$\bo(2d_1^4, d_1, d_1, 0)$} but before completing Procedure~$\p(d_1)$, {$a_2$} performs {$\rd(2d_1^4+3d_1, C(\bo(2d_1^4, d_1, d_1, 0)))$}. {In view of the previous paragraph, the execution of this pattern by $a_2$ is preceded by the execution of $\bo(2d_1^4, d_1, d_1, 0)$ by $a_1$.} When {$a_2$} finishes executing $\rd(2d_1^4+3d_1, C(\bo(2d_1^4, d_1, d_1, 0)))$, in view of Lemma~\ref{lem:rd:bo} and since by assumption the agents have not met, {$a_1$} has finished executing Pattern~$\bo(2d_1^4, d_1, d_1, 0)$. This means that when {$a_2$} finishes executing $\p(d_1)$ and thus the second part of Phase~$i_1$, {$a_1$} has completed the execution of the first part of Phase~$i_1$, which proves the lemma.
		\end{proof}

		In the following lemma, we focus on the calls to Pattern $\rd$ in the second and in the third part of Procedure $\ite(d_1)$ for any power of two $d_1$. In the statement and proof of this lemma, they are called ``synchronization $\rd$'', and indexed from 1 to $(d_1 (2d_1 (d_1 + 1) + 1) + 1)$ in their ascending execution order in these two parts of the procedure. During any execution of Procedure $\ite(d_1)$ for any power of two $d_1$, the call to Pattern $\rd$ in the second part of Procedure $\ite$ is the first (indexed by 1) synchronization $\rd$ of this procedure.

		\begin{lemma} \label{lemmaBitSynchronization}
			Let $a_1$ and $a_2$ be two agents executing Algorithm \aud{}. Let $v_1$ and $v_2$ be their respective initial nodes separated by a distance $\dist$. For any power of two $d_1 \geq \dist$ and any positive integer $i \leq d_1 (2d_1 (d_1 + 1) + 1) + 1$, if the agents have not met yet, then when any of them completes the execution of the $i$-th synchronization $\rd$ of $\ite(d_1)$, the other agent has at least started it.
		\end{lemma}

		\begin{proof}
%			Consider two agents $a_1$ and $a_2$ respectively initially located on nodes $v_1$ and $v_2$ separated by a distance $\dist$.

			Suppose that agent {$a_2$} has just finished executing the $i$-th synchronization $\rd$ inside Procedure $\ite(d_1)$ for any power of two $d_1 \geq \dist$ and any positive integer $i \leq d_1 (2d_1 (d_1 + 1) + 1) + 1$. Let us prove by induction on $i$ that if rendezvous has not occurred yet then {$a_1$} has at least started executing this $i$-th synchronization $\rd$.

			Let us first consider the case in which $i = 1$. The synchronization $\rd$ {$a_2$} has just
			finished executing is called at the end of the execution of Procedure $\p(d_1)$ called at
			line~\ref{MainFireworks'Call} of Procedure~$\ite(d_1)$. Since $d_1 \geq \dist$, in view of
			Lemma~\ref{lemmaFireworks}, when {$a_2$} completes the execution of the first synchronization $\rd$ and thus the execution of $\p(d_1)$, either the agents have met or {$a_1$} has completed the execution of the first part of Procedure $\ite(d_1)$ \ie begun the execution of the first synchronization $\rd$.

			Let us now make the assumption that for any power of two $d_1 \geq \dist$, during any execution of Procedure $\ite(d_1)$, there exists an integer $j$ {from 1 to $d_1 (2d_1 (d_1 + 1) + 1) + 1$} such that when agent {$a_2$} completes the execution of the $j$-th synchronization $\rd$, either the agents have met or {$a_1$} has at least started the execution of the $j$-th synchronization $\rd$, and prove that when {$a_2$} completes the execution of the $(j + 1)$-th synchronization $\rd$, either the agents have met or {$a_1$} has at least started the execution of the same synchronization $\rd$. Let us assume by contradiction that when {$a_2$ finishes} executing the $(j + 1)$-th synchronization $\rd$, {$a_1$} has neither met {$a_2$} nor started executing the $(j + 1)$-th synchronization $\rd$.

			After executing the $j$-th synchronization $\rd$, {$a_2$} executes line~\ref{callFL} or line~\ref{callBO} of Algorithm $\ite(d_1)$ and thus either Pattern $\fl$ or Pattern $\bo$, depending on the bits of its transformed label. {The induction hypothesis implies that the execution of the $j$-th synchronization $\rd$ by $a_1$ precedes the execution by $a_2$ of either $\fl$ or $\bo$ between the $j$-th and the $(j + 1)$-th synchronisation $\rd$}. In view of Lemmas~\ref{lem:fl:di} and~\ref{lem:bo}, as $d_1 \geq D$, whichever pattern {$a_2$} executes, it pushes the execution of the $j$-th synchronization $\rd$ by {$a_1$}. By assumption, when {$a_2$} finishes executing line~\ref{callFL} or line~\ref{callBO} of Algorithm $\ite(d_1)$ after the $j$-th synchronization $\rd$, the agents have not met which implies that {$a_1$} has finished executing the $j$-th synchronization $\rd$.

			The next pattern that {$a_2$} executes is the $(j + 1)$-th synchronization $\rd$. Given the above assumptions and statements, when {$a_2$} starts executing this synchronization $\rd$, {$a_1$} has finished executing the $j$-th synchronization $\rd$ and has started executing line~\ref{callFL} or line~\ref{callBO} of Algorithm $\ite(d_1)$. In view of Lemmas~\ref{lem:rd:fl} and~\ref{lem:rd:bo}, since $d_1 \geq \dist$, whichever pattern {$a_1$} executes, it is pushed by the execution of the $(j + 1)$-th synchronization $\rd$ by {$a_2$}. Given that, still by assumption, the agents do not meet before {$a_2$} completes the execution of the $(j + 1)$-th synchronization $\rd$, when this occurs, {$a_1$} has completed the execution of line~\ref{callFL} or~\ref{callBO} of Algorithm $\ite(d_1)$, just after the $j$-th, and just before the $(j + 1)$-th synchronization $\rd$. Hence, when {$a_2$} completes the execution of the $(j + 1)$-th synchronization $\rd$, {$a_1$} has at least started executing the $(j + 1)$-th synchronization $\rd$, which contradicts the hypothesis that when {$a_2$} completes the execution of the $(j + 1)$-th synchronization $\rd$, {$a_1$} has neither met {$a_2$} nor started executing the $(j + 1)$-th synchronization $\rd$, and proves the lemma.
		\end{proof}

	\subsection{Correctness of Algorithm \aud{}}

		\begin{theorem} \label{lemma6Cases}
			Algorithm \aud{} solves the problem of rendezvous in the basic grid.
		\end{theorem}

		\begin{proof}
			To prove this theorem, it is enough to prove the following claim.

			\begin{claim} \label{claim6Cases}
				Let $d_1$ be the smallest power of two such that $d_1 \geq max(\dist, \difbit)$ with $\difbit$ the index of the first bit which differs in the transformed labels of the agents.
				Algorithm \aud{} ensures rendezvous by the time any agent completes an execution of Procedure $\ite(d_1)$.
			\end{claim}

			First, in view of Remark~\ref{alg:rem}, $\difbit$ exists. Respectively denote by $v_1$ and $v_2$, the initial nodes of {two agents denoted by} $a_1$ and $a_2$. This proof is made by contradiction. Suppose that the agents $a_1$ and $a_2$ {execute Algorithm \aud{} but do not meet by the time any agent completes an execution of Procedure $\ite(d_1)$ where $d_1$ is the smallest power of two such that $d_1 \geq max(\dist, \difbit)$}.
	
			{This in particular means that one of the agents eventually starts executing $\ite(d_1)$}. Since $d_1 \geq \dist$, in view of Lemma~\ref{lemmaFireworks}, we know that as soon as this agent completes the execution of Procedure $\p(d_1)$, both agents have started executing $\ite(d_1)$. Otherwise, agents have met which contradicts our assumption. Without loss of generality, suppose that the bits in the transformed labels of agents $a_1$ and $a_2$ with the index $\difbit$ are respectively 1 and 0.

			In order to prove this claim, we first show that there exists an iteration of the loop at line~\ref{firstLoop} of Algorithm~\ref{iteration} during which the two following properties are satisfied:
			\begin{enumerate}
				\item the value of variable $i$ is equal to $\difbit$
				\item the value of variable $j$ is such that when executing Pattern $\bo$ at line~\ref{callBO}, the first pair of Patterns $\di$ and $\fl$ executed inside this $\bo$ by $a_1$ starts from $v_2$
			\end{enumerate}

			The first property follows from the fact that $d_1 \geq \difbit$.

			We now show that the second property is also satisfied. Let $U$ be a list of all the nodes at
			distance at most $d_1$ from $v_1$ and ordered in the order of the first visit when executing $\di(d_1)$ from node $v_1$. The same list is considered in the algorithm of Pattern $\bo(\dip, d_1, d_1, \bops)$ for any positive integers $\dip$ and $\bops$. First of all, there are $2d_1 (d_1 + 1) + 1$ nodes at distance at most $d_1$ from $v_1$, and thus in $U$.
% Indeed, for every integer $i$ between 1 and $d_1$, there are $4i$ nodes at distance $i$ from $u$. Thus, there are $\sum_{j = 1}^{d_1} {(4j)}$ nodes at a distance between 1 and $d_1$ from $u$, which is equal to $2d_1 (d_1 + 1)$. By adding $u$ to these $2d_1 (d_1 + 1)$ nodes, we obtain $2d_1 (d_1 + 1) + 1$ nodes.
Since the distance between $v_1$ and $v_2$ is $\dist \leq d_1$, $v_2$ belongs to $U$. Denote by $j_1$ its index (between 0 and $2d_1 (d_1 + 1)$) in $U$. According to Procedure $\ite$, the value of variable $j$ is
incremented at each iteration of the loop at line~\ref{firstLoop} and takes one after another each integer value between 0 and $2d_1 (d_1 + 1)$. Consider the iteration when it is equal to $j_1$. According to Algorithm $\bo$, the first node from which $a_1$ executes $\di$ and $\fl$ is the node which has index $j_1 + 0 \pmod{2d_1 (d_1 + 1) + 1} = j_1$. This node is $v_2$, which proves that there exists an iteration of the loop at line~\ref{firstLoop} during which the second property is verified too. Let us denote it by $I$. It is the iteration after the $(1 + (\difbit - 1) (2d_1 (d_1 + 1) + 1) + j_1)$-th synchronization $\rd$ {inside Phase~$d_1$}.
	
			In view of Lemma~\ref{lemmaBitSynchronization}, we know that when an agent completes its execution of the $i$-th synchronization $\rd$ inside the second and the third part of any execution of Procedure $\ite(d_1)$ (for any positive integer $i$ {less} than or equal to $(d_1 (2d_1 (d_1 + 1) + 1) + 1$), the other agent has at least begun the execution of this synchronization $\rd$. Thus, when an agent is the first one which starts executing $I$, it has just finished executing the $(1 + (\difbit - 1) (2d_1 (d_1 + 1) + 1) + j_1)$-th synchronization $\rd$ and the other agent is executing (or finishing executing) the same $\rd$. Let us prove that rendezvous occurs before any of the agents starts the next synchronization $\rd$.

			Let us consider the patterns both agents execute between the beginning of the $(1 + (\difbit - 1) (2d_1 (d_1 + 1) + 1) + j_1)$-th
			synchronization $\rd$, and the beginning of the next one. Agent $a_1$ executes Pattern
			$\rd(\dip, n)$ with $\dip$ and $n$ two positive integers (call this pattern $p_1$)
			and Pattern $\bo(\dip, d_1, d_1, j_1)$ from node $v_1$ while $a_2$ executes $\rd(\dip, n)$ (let
			us call it $p_2$) and $\fl(\dip, d_1)$ (this is $p_3$) from node $v_2$. During its execution of Pattern $\bo(\dip, d_1, d_1, j_1)$ from node $v_1$, $a_1$ first follows $\pa(v_1, v_2)$, and then executes Pattern $\di(\dip)$ followed by Pattern $\fl(\dip, d_1)$ both from node $v_2$ (call them respectively $p_4$ and $p_5$). Recall that during any execution of Pattern $\fl(\dip, d_1)$ from node $v_2$, there are two periods, the second one consisting in backtracking every edge traversal made during the first one. During the first period, in particular, an agent executes a Pattern $\di(\dip)$ from every node at distance at most $d_1$, among which there are node $v_1$ and node $v_2$. Since backtracking $\di(\dip)$ allows to perform exactly the same edge traversals as $\di(\dip)$, during the second period of Pattern $\fl(\dip, d_1)$, there is also an execution of Pattern $\di(\dip)$ from node $v_1$ and another from $v_2$.

			Let us consider two different cases. In the first one, when $a_1$ starts executing $p_4$ from $v_2$, inside $p_3$, $a_2$ has not yet started following $\pa(v_2, v_1)$ to go executing $\di(\dip)$ from $v_1$. In the second one, when $a_1$ starts executing $p_4$ from $v_2$, $a_2$ has at least started following $\pa(v_2, v_1)$ to go executing $\di(\dip)$ from $v_1$. In the following, we analyse both these cases.

			In the first case, consider what $a_2$ can be executing when
			$a_1$ starts executing $p_4$ from node $v_2$ after following $\pa(v_1, v_2)$. First, it can still
			be executing the synchronization $\rd$ $p_2$ from node $v_2$. Then, in view of Lemma~\ref{lem:di:rdv},
			rendezvous occurs. The only other pattern that $a_2$ can be executing at this moment is $p_3$. However, in this case, we know that $a_2$ will have finished its execution of $p_3$ before $a_1$ starts $p_5$, just after $p_4$. Otherwise, in view of Lemma~\ref{lem:fl:rdv}, rendezvous occurs.
	
			We have just reminded the reader that during any execution of Pattern $\fl(\dip, d_1)$ from $v_2$, agent $a_2$ performs, among the patterns $\di(\dip)$ from every node at distance at most $d_1$ from $v_2$, {Pattern} $\di(\dip)$ from $v_2$. If it executes one of these Patterns $\di(\dip)$ while $a_1$ is executing its $p_4$ from node $v_2$ after following $\pa(v_1, v_2)$, in view of Lemma~\ref{lem:di:rdv}, rendezvous occurs. This implies that before $a_1$ finishes following $\pa(v_1, v_2)$, $a_2$ has completed each execution of Pattern $\di(\dip)$ from $v_2$ inside its execution of $\fl(\dip, d_1)$.
	
			This means that, each execution of Pattern $\di(\dip)$ from node $v_2$ during the second period of $p_3$ has already been completed by $a_2$ when $a_1$ starts executing its own $\di(\dip)$ from $v_2$. Since inside the second period of $p_3$, $a_2$ executes Pattern $\di(\dip)$ from node $v_2$, $a_2$ has already executed the whole first period of $p_3$ when $a_1$ starts executing $p_4$ from $v_2$ including Pattern $\di(\dip)$ performed from node $v_1$, since $v_1$ is at distance at most $d_1$ from $v_2$. This contradicts the definition of this first case: according to this definition, when $a_1$ starts executing $p_4$ from $v_2$, inside $p_3$, $a_2$ has not followed $\pa(v_2, v_1)$ yet, and thus has not executed $\di(\dip)$ from $v_1$.
	
			In the second case, we prove that rendezvous occurs, which is a contradiction. Recall that in this case, when $a_1$ starts executing $p_4$ from $v_2$, $a_2$ has at least started following $\pa(v_2, v_1)$ to go executing $\di(\dip)$ from $v_1$. If $a_2$ has not finished following $\pa(v_2, v_1)$ when $a_1$ starts following $\pa(v_1, v_2)$, then if we denote by $t_1$ (resp. $t_2$) the time when $a_1$ (resp. $a_2$) finishes following $\pa(v_1, v_2)$ (resp. $\pa(v_2, v_1)$), agents meet by time $min(t_1, t_2)$ since $\pa(v_1, v_2) = \overline{\pa(v_2, v_1)}$. If $a_2$ has finished following $\pa(v_2, v_1)$ before $a_1$ starts executing $\pa(v_1, v_2)$, then it has started executing $\di(\dip)$ from $v_1$ before $a_1$ finishes executing $p_1$ (before it executes $\bo(\dip, d_1, d_1, j_1)$), which means in view of Lemma~\ref{lem:di:rdv} that the agents achieve rendezvous.
	
			So, whatever the execution chosen by the adversary, rendezvous occurs in the worst case by the time any agent completes $\ite(d_1)$, which proves the claim, and by extension the theorem.
		\end{proof}

	\subsection{Cost analysis}

		\begin{theorem} \label{theo:cost}
			The cost of Algorithm \aud{} {belongs to $\mathrm{O}((\dist+\sholab)^{33})$.}
		\end{theorem}

		\begin{proof}
			In order to prove this theorem, we first need to show the following two claims.

			\begin{claim} \label{cost:pat}
				The cost of each basic pattern inside $\mathcal{BD}(\ite(d_1))$ (with $d_1$ any power of two) {is in $\mathrm{O}(d_1^{30})$.}
			\end{claim}

{			To prove this claim, we are going to exhibit the most costly basic pattern which could belong
	to $\mathcal{BD}(\ite(d_1))$ and prove that its cost is in $\mathrm{O}(d_1^{30})$.
			
			From their algorithms, we get the following upper bounds on the costs of our various basic patterns: $C(\di(\dip)) \in \mathrm{O}(\dip ^ 2)$, $C(\rd(\dip, n)) \in \mathrm{O}(n \times C(\di(\dip)))$, $C(\fl(\dip, \flp)) \in \mathrm{O}((\dip + \flp) ^ 5)$, and $C(\bo(\dip, \flp, \bop, \bops)) \in \mathrm{O}(\bop ^ 2 \times (z + C(\di(\dip)) + C(\fl(\dip, \flp))))$. Remark that the higher the values of their parameters are, the higher the costs of our patterns are (except for the fourth parameter of $\bo$ which does not impact its cost).
			
			Also notice that pattern $\di$ does not belong to $\mathcal{BD}(\ite(d_1))$. It is called when executing the other basic patterns. {Moreover}, if they are given the same values for their two first parameters, Pattern~$\bo$ is more costly than $\fl$, which makes it a good candidate for being our most costly pattern. But, when called with a second parameter which is the cost of $\bo$, Pattern~$\rd$ is much more costly than the latter which makes it the most costly pattern inside $\mathcal{BD}(\ite(d_1))$. Remark that in our algorithm, the second parameter of Pattern~$\rd$ is the cost of either $\fl$ or $\bo$. In particular, it cannot be the cost of another $\rd$.
			
			In view of Lemma~\ref{cor:syn:ite}, for any power of two $d_1$, inside $\mathcal{BD}(\ite(d_1))$, the value of the first parameter given to our patterns is at most $32d_1^4 - 6d_1$. {In addition}, for each basic pattern $\fl$ or $\bo$ inside $\mathcal{BD}(\ite(d_1))$, the value given to its second parameter is always $d_1$. This gives us upper bounds on the values of the parameters our most costly pattern can be given. Hence, the cost of each pattern we call inside $\mathcal{BD}(\ite(d_1))$ is at most $C(\rd(32d_1^4 - 6d_1, C(\bo(32d_1^4 - 6d_1, d_1, d_1, \bops))))$ with $\bops$ any positive integer. This cost belongs to $\mathrm{O}(C(\rd(d_1^4, d_1^2(d_1 + (d_1^4)^2 + (d_1^4)^5))))$ and thus to $\mathrm{O}((d_1^4)^2d_1^{22})$ \ie to $\mathrm{O}(d_1 ^{30})$.}		

			\begin{claim} \label{cost:ite}
				The cost of Procedure $\ite(d_1)$ (with $d_1$ any power of two) {belongs to $\mathrm{O}(d_1^{33})$.}
			\end{claim}

			In view of Definition~\ref{def:decompo} and Remark~\ref{rem:per}, for any power of two $d_1$, the cost of Procedure $\ite(d_1)$ is the same as the sum of the costs of all the basic patterns inside $\mathcal{BD}(\ite(d_1))$. In view of Claim~\ref{cost:pat}, we know that for any power of two $d_1$, inside $\mathcal{BD}(\ite(d_1))$, {the cost of each basic pattern is in $\mathrm{O}(d_1^{30})$.} Thus, to prove this claim it is enough to show that $\mathcal{BD}(\ite(d_1))$ contains a number of basic patterns which {is in $\mathrm{O}(d_1^3)$.}

			For any power of two $d_1$, Procedure $\ite(d_1)$ is composed of a call to $\p(d_1)$ and the nested loops. These loops consist in $2d_1 (2d_1 (d_1 + 1) + 1)$ calls to basic patterns. Half of them are made to $\rd$ and the others either to $\fl$ or to $\bo$. In its turn, $\p(d_1)$ is composed of two parts: a loop calling Procedure $\pp$ and two basic patterns. For any power of two $d_2$, in view of Algorithm~\ref{patternPP}, and since they are both perfect, the number of basic patterns inside $\mathcal{BD}(\pp(d_2, d_1))$ or $\mathcal{BD}(\ite(d_2))$ is the same. As a consequence, if $d_1 \geq 2$, $\mathcal{BD}(\pp(1, d_1))$, \ldots, $\mathcal{BD}(\pp(\frac {d_1} {2}, d_1))$ is composed of as many basic patterns as there are in $\mathcal{BD}(\ite(1))$, \ldots, $\mathcal{BD}(\ite(\frac {d_1} {2}))$.

			For any power of two $i$, let us denote by $L_1(i)$ (resp. $L_2(i)$) the number of calls to basic patterns inside $\mathcal{BD}(\ite(i))$ (resp. $\mathcal{BD}(\p(i))$). We then have the following equations:
			\begin{gather*}
				L_1(i) = L_2(i) + 2i (2i (i + 1) + 1) \\
				L_2(i) = \sum_{j = 0}^{\log_2(i) - 1} (L_1(2 ^ j)) + 2
			\end{gather*}
			They imply the following:
			\begin{gather*}
				L_2(1) = 2 \quad \text{and} \\
				\text{if} \quad i \geq 2 \quad \text{then} \quad L_2(i) = L_2(\frac {i} {2}) + L_1 (\frac {i} {2}) = 2L_2(\frac {i} {2}) + i( i (\frac {i} {2} + 1) + 1)
				\intertext{which we can also write}
				L_2(i) = 2i + \sum_{j = 1}^{\log_2(i)} (2 ^ {\log_2(i) - j} \cdot 2 ^ j ( 2 ^ j ( 2 ^ {j - 1} + 1) + 1)) \\
				L_2(i) = 2i + i\sum_{j = 1}^{\log_2(i)} (2 ^ j (2 ^ {j - 1} + 1) + 1) \\
				L_2(i) = 2i + i\sum_{j = 1}^{\log_2(i)} (2 ^ {2j - 1} + 2 ^ j + 1) \\
				L_2(i) = 2i + i(\log_2(i) + \frac {2 (i ^ 2 - 1)} {3} + 2(i - 1))
			\end{gather*}
			Hence, {both $L_2(i)$ and $L_1(i)$ belong to $\mathrm{O}(i^3)$. This means that for any power of two $d_1$, $\mathcal{BD}(\ite(d_1))$ is composed of a number of basic patterns which is in $\mathrm{O}(d_1^3)$. Hence, in view of Claim~\ref{cost:pat}, the cost of $\ite(d_1)$ indeed belongs to $\mathrm{O}(d_1^{33})$}, which proves the claim.

			Now, it remains to conclude the proof of the theorem. In view of Claim~\ref{claim6Cases}, rendezvous is achieved by the end of the execution of $\ite(\delta)$ by any of the agents, where $\delta$ is the smallest power of two such that $\delta \geq max(\dist, \difbit)$ and $\difbit$ is the index of the first bit which differs in the transformed labels of the agents. {Moreover, in view of Claim~\ref{cost:ite}, the cost of each call to $\ite(d_1)$ for some power of two $d_1 \leq \delta$ belongs to $\mathrm{O}(d_1^{33})$. Since $\sum_{i=0}^{\log\delta}({2^i}^{33})\leq 2\delta^{33}$, the sum of the costs of these calls to Procedure~$\ite$ and thus the cost of our algorithm until rendezvous is achieved belongs to $\mathrm{O}(\delta^{33})$. {Moreover}, by construction, we have $\difbit\leq2\sholab+2$. This means that the cost of our algorithm belongs to $\mathrm{O}((\dist+\sholab)^{33})$.}

			%The index $\difbit$ is smaller than the length of the shorter transformed label. Besides, given the way this transformed label is built, its length is $2\sholab + 2$ which implies that $\difbit \leq 2\sholab + 2$. At last, the value of variable $d$ only changes (it doubles) after an execution of procedure $\ite$ during which rendezvous does not occur. This means that $\delta < 2 max(\dist, 2\sholab + 2)$: it is polynomial in $\dist$ and $\sholab$.

			%Hence, in view of Claim~\ref{cost:ite}, for any power of two $d_1 \leq \delta$, the cost of procedure $\ite(d_1)$ is polynomial in $\dist$ and $\sholab$. Besides, since the value of $d$ doubles after each execution of procedure $\ite$, for each agent, there are $\log_2 (\delta) + 1$ (which is indeed polynomial in $\dist$ and $\sholab$) executions of procedure $\ite$. The cost of Algorithm \aud{} is polynomial in $\dist$ and $\sholab$.
		\end{proof}

\section{Conclusion}
\label{sec:ccl}

From Theorems~\ref{theo:preli},~\ref{lemma6Cases} and~\ref{theo:cost}, we obtain the following result concerning the task of approach in the plane.

\begin{theorem}
The task of approach can be solved at cost polynomial in the unknown initial distance $\Delta$ separating the agents
and in the length of (the binary representation) of the shortest of their labels.
\end{theorem}

Throughout the paper, we made no attempt at optimizing the cost. Actually, as the {attentive reader} will have
noticed, our main concern was only to prove the polynomiality. Hence, a natural open problem is to find out the
optimal cost to solve the task of approach. This would be all the more important as in turn we could compare this
optimal cost with the cost of solving the same task with agents that can position themselves in a global system of
coordinates (the almost optimal cost for this case is given in~\cite{CollinsCGL10}) in order to determine whether the
use of such a system (\eg GPS) is finally relevant to minimize the travelled distance.

	\bibliographystyle{plain}

\end{document}